\documentclass[sn-basic, iicol]{sn-jnl}% Math and Physical Sciences Numbered Reference Style 
\usepackage{graphicx}%
\usepackage{multirow}%
\usepackage{amsmath,amssymb,amsfonts}%
\usepackage{amsthm}%
\usepackage{mathrsfs}%
\usepackage{appendix}%
\usepackage{xcolor}%
\usepackage{textcomp}%
\usepackage{manyfoot}%
\usepackage{booktabs}%
\usepackage{algorithm}%
\usepackage{algorithmicx}%
\usepackage{mathtools}%
\usepackage{algpseudocode}%
\usepackage{listings}%
\usepackage{lmodern}%
\counterwithout{figure}{section}
\usepackage{bbm}%
\usepackage{subcaption}
\usepackage{diagbox}
\DeclareMathOperator*{\argmax}{arg\,max}

\DeclareMathOperator*{\maxlen}{L_\text{max}}

\newcommand{\defeq}{\vcentcolon=}

\newtheorem{theorem}{Theorem}%  meant for continuous numbers
\newtheorem{proposition}[theorem]{Proposition}% 

\newtheorem{definition}{Definition}%

\begin{document}

\title[A novel framework for quantifying nominal outlyingness]{A novel framework for quantifying nominal outlyingness}

%%=============================================================%%
%% GivenName	-> \fnm{Joergen W.}
%% Particle	-> \spfx{van der} -> surname prefix
%% FamilyName	-> \sur{Ploeg}
%% Suffix	-> \sfx{IV}
%% \author*[1,2]{\fnm{Joergen W.} \spfx{van der} \sur{Ploeg} 
%%  \sfx{IV}}\email{iauthor@gmail.com}
%%=============================================================%%

\author*[1]{\fnm{Efthymios} \sur{Costa}}\email{efthymios.costa17@imperial.ac.uk}
\equalcont{These authors contributed equally to this work.}

\author[1]{\fnm{Ioanna} \sur{Papatsouma}}\email{i.papatsouma@imperial.ac.uk}
\equalcont{These authors contributed equally to this work.}

\affil*[1]{\orgdiv{Department of Mathematics}, \orgname{Imperial College London}}

%%==================================%%
%% Sample for unstructured abstract %%
%%==================================%%

\abstract{Outlier detection is an important data mining tool that becomes particularly challenging when dealing with nominal data. First and foremost, flagging observations as outlying requires a well-defined notion of nominal outlyingness. This paper presents a definition of nominal outlyingness and introduces a general framework for quantifying outlyingness of nominal data. The proposed framework makes use of ideas from the association rule mining literature and can be used for calculating scores that indicate how outlying a nominal observation is. Methods for determining the involved hyperparameter values are presented and the concepts of variable contributions and outlyingness depth are introduced, in an attempt to enhance interpretability of the results. The proposed framework is evaluated on both synthetic and publicly available data sets, demonstrating comparable performance to state-of-the-art frequent pattern mining algorithms and even outperforming them in certain cases. The ideas presented can serve as a tool for assessing the degree to which an observation differs from the rest of the data, under the assumption of sequences of nominal levels having been generated from a Multinomial distribution with varying event probabilities.}


%%\pacs[JEL Classification]{D8, H51}

%%\pacs[MSC Classification]{35A01, 65L10, 65L12, 65L20, 65L70}

\maketitle

\section{Introduction}\label{sec:intro}

%Outlier detection aims to flag atypical observations in a data set; these are the `outliers' (also called anomalies) and they may be points with values that do not conform to some pattern suggested by the data or observations whose values arouse suspicion regarding the mechanism that has been used to generate them.
Outlier detection (also known as anomaly detection) is a crucial aspect of data analysis that has received increasing attention by data mining researchers. The wide range of uses of outlier detection in various domains, such as in cybersecurity \citep{aggarwal2007data, di2008intrusion}, sociology \citep{savage2014anomaly}, healthcare \citep{tschuchnig2022anomaly}, and finance \citep{ngai2011application}, has driven the development of numerous algorithms designed to detect anomalies within a data set. Most of these algorithms are restricted to just one data type and that is continuous data. However, categorical data, that is binary, ordinal and nominal data \citep{kaufmanrousseeuw2009}, is prevalent in real-world data sets, yet very few outlier detection methods can handle this data type and authors do not always make the distinction between purely nominal and ordinal variables. Outliers in data sets with nominal variables may indicate a serious offence or a critical situation. For instance, unusual combinations of symptoms, diagnoses and demographic characteristics of patients could be indications of rare conditions in healthcare monitoring or uncommon combinations of claim types, policies and accident locations may hint fraud in insurance claims. This demands the development of outlier detection methods tailored for nominal data.

In this paper, we focus on nominal variables, including binary features. For completeness, we present some of the related literature below. Since no standard method exists, we provide a brief overview of the various approaches that have been proposed. For an extensive review of outlier detection algorithms for categorical data, in the sense of both nominal and ordinal data, we refer the reader to \cite{taha2019anomaly}.

%In multivariate outlier detection for continuous data, the main idea is to identify which observations are less likely to occur under the assumption of a model having generated the data. This is described in more detail in \cite{becker1999masking}, who look more specifically at the case of the mutivariate Gaussian distribution. The outlier identification problem is commonly approached by three steps. Firstly, the parameters of the underlying distribution are estimated in a robust way (that is, without being influenced by observations that are likely outlying), then robust distances are computed using the robust estimators and finally, the distances are compared to a cutoff value. This cutoff value is typically the $(1-\alpha)$ quantile of the distribution of the distances, assuming the hypothesised underlying model is true. Examples of robust estimators include the well-celebrated Minimum Covariance Determinant (MCD) \citep{rousseeuw1984least}, the Stahel-Donoho estimator \citep{stahel1981robuste, donoho1982breakdown} or S-estimators \citep{rousseeuw1984robust}, among others. Notice how these assume an elliptical distribution having generated the data, thus allowing for the use of Mahalanobis-type distances in order to detect outliers.

One approach to outlier identification assuming an underlying discrete distribution, is extending ideas from the continuous to the discrete domain. This would typically require robust estimation of the parameters of some assumed underlying distribution, computation of robust distances and finally comparing the distances to a cutoff value \citep[see][for examples of robust estimators]{stahel1981robuste, rousseeuw1984least, rousseeuw1985multivariate}. The cutoff value is typically the $1-\alpha$ quantile of the distribution of the distances, assuming validity of the hypothesised data generating mechanism. In the discrete distribution setting, this is generally a hard task that is not possible unless some model is assumed to have generated the counts in a contingency table. Log-linear models are widely used for this purpose and one can then infer what the expected proportions should be in a robust way \citep[see for example][]{kuhnt2010breakdown, kuhnt2014outlier, calvino2021robustness}. However, these methods can become impractical when the number of nominal variables is very large and the analysis of a high-dimensional contingency table is required. The sparsity of the resulting contingency tables increases together with the dimensionality, which is something that needs to be accounted for as well.

Another notable approach to outlier detection for nominal data is using \textit{proximity-based} methods. These methods seek to detect outliers by generating a distance matrix between observations and finding the subjects that are furthest away from the rest of the points in the data set. Examples of such methods include \cite{bay2003mining} and \cite{li2007mining}. However, distances in the categorical space often lack meaningful interpretation and implicitly make an equidistance assumption. As an example, the Hamming distance treats all mismatches equally. \cite{sripriya2020robust} introduced a distance metric for detecting outliers in categorical data; however, its application is restricted to $2 \times 2$ contingency tables. Furthermore, these methods are very impractical in the presence of a large number of nominal variables, leading to a high computational cost.

\textit{Information-theoretic} methods have also been proposed for outlier identification of nominal data \citep{lee2000information, wu2011information, zhao2014simple}. Such methods rely on some information-theoretic measure of \textit{data disorder}, such as the \textit{entropy}, and typically employ a local search algorithm to optimise an objective function. As a result, these methods rely heavily on heuristics and the optimisation process can be very costly. These problems are also common in recently proposed deep learning methods \citep[for example][]{chen2016, cheng2019neural}.

One different approach to anomaly detection for nominal data is the \textit{rule-based} approach that makes use of ideas from the association rule mining literature \citep{agrawal1993mining}. To the best of our knowledge, the first algorithm that adopted this mindset for dealing with the problem of outlier detection in nominal data was the Link-based Outlier and Anomaly Detection in Evolving Data sets (LOADED) algorithm of \cite{ghoting2004loaded}. In fact, LOADED was developed for identifying anomalies in mixed-type data (that is for a combination of nominal and continuous features). An improvement of this algorithm in terms of computational cost was then presented by \cite{otey2006fast}, with further ameliorations being proposed by \cite{koufakou2010} for the nominal outlier identification part. More precisely, \cite{koufakou2010} introduced the Outlier Detection for Mixed Attribute Datasets (ODMAD) algorithm. The primary concern with the way ODMAD deals with nominal features is that it treats outliers in the nominal domain in a way that disregards the dimension of the subspace in which these are found. Furthermore, the implementation of ODMAD involves certain threshold parameters, the values of which need to be defined by the user, without providing guidelines of how these could be chosen in a data-driven manner.

Similarly, several frequent pattern mining methodologies have been employed for quantifying nominal outlyingness, such as the Frequent Pattern Outlier Factor (FPOF) \citep{he2004frequent} and the Frequent Pattern Isolation (FPI) \citep{kuchar2018spotlighting} algorithms. However, these methods also suffer from a lack of intuition regarding the choice of their associated hyperparameter values and do not take into account that longer sequences are inherently less frequent than their subsets. Moreover, they completely overlook whether an event is rare by nature. For instance, consider a rare disease that affects only a small proportion of the population but is known and not particularly noteworthy. Frequent pattern mining algorithms will most likely flag subjects suffering from that disease as potential outliers solely due to their low frequency within a data set of clinical records, even though these are completely normal and expected cases. This underscores the need for a framework that can incorporate such contextual information and produce results that align with the underlying ground truth.

%The main contributions of this work are as follows: we improve upon existing approaches to the problem of outlier detection in the categorical domain, leveraging key insights that stem from prior research in the field, borrowing ideas from the association rule mining literature and further reducing the amount of required user input throughout the process. In this paper we propose a novel framework for quantifying the outlyingness of nominal data.
The rest of the paper is organised as follows: we define the concept of nominal outlyingness by adopting a rule-based mindset in Section \ref{sec:defoutliers} and we give a brief introduction to key concepts from the area of association rule mining. We then propose a score that quantifies abnormal behaviour in Section \ref{sec:proposal}. Our main contribution is the formulation of a score that takes into account any contextual information, relative proximity of a nominal observation to infrequency, and sparsity of the contingency tables involving multiple nominal features. Based on these considerations, we provide guidelines on the selection of the associated hyperparameters and introduce two concepts that can be used for meta-analysis purposes; these are described in Section \ref{sec:proposal} as well. Some mathematical properties of our proposed score are included in Section \ref{sec:properties}. The efficacy of our method is illustrated via a series of simulations on synthetic and publicly available data sets in Section \ref{sec:applications}, where our proposal is benchmarked against two state-of-the-art frequent pattern mining algorithms. We finally make some concluding remarks in Section \ref{sec:conclusion}.

\section{Outlyingness in discrete space}\label{sec:defoutliers}

In this Section we extend concepts of multivariate outlier detection for continuous data to discrete spaces. We define the concept of nominal outlyingness and give a brief introduction to basic notions of association rule mining.

\subsection{Definition of nominal outlyingness}

%\textcolor{red}{One may find several definitions of outlyingness for continuous observations in the literature.}
Outlyingness for continuous data is typically defined based on assumptions related to the data generating mechanism \citep{davies1993identification, becker1999masking}, or the location of each observation relative to the rest of the points in the feature space \citep[e.g.][]{breunig2000lof,liu2008isolation}. 
%\textcolor{red}{If nominal variables appear in a data set, using the location of points becomes difficult to work with and looking at distances in the categorical space is not ideal for reasons previously stated in Section \ref{sec:intro}. As we are dealing with point masses instead of densities, only a fixed number of values can occur, meaning that location may no longer be used as a criterion to assess anomalous behaviour. However, one can make use of the ideas related to the identification of outliers with respect to the mechanism that has been used to generate the data. In fact, all we have to assume is that the data has been generated from some underlying discrete distribution. Then, we turn our attention to observations of abnormally high or low mass with respect to this discrete distribution, just like we would suspect points of low density to be potential outliers in the continuous setting. This leads to two possible definitions of nominal outlyingness; one for observations whose nominal values are too frequent and one for too infrequent values, always with respect to an assumed discrete random variable. The notion of \textit{frequency} is described in detail later in the section.}
However, location and distance are not meaningful criteria for detecting anomalies among nominal features. Assuming an underlying discrete distribution, we define outliers as values with unusually high or low frequency, analogous to low-density points in the continuous case.

Notice that looking at both highly frequent and highly infrequent values is rather unsound; 
%\textcolor{red}{this is easy to realise by considering a toy example of one binary variable generated by a Bernoulli distribution for which the success probability has been significantly misspecified. In such a case, one of the two values will be treated as highly infrequent. However, this implicitly means that the other value will be highly frequent.}
for instance, a highly infrequent level of a binary variable renders the other level highly frequent by default. Therefore, one needs to specify whether their interest lies in highly frequent or highly infrequent values, that is, what they wish to define as a nominal outlier. This leads us to the following definition of nominal outlyingness:

%\textcolor{red}{Extending these ideas to multiple nominal features requires taking combinations of their values into account. One potential downside of such an approach is that in the presence of multiple nominal variables taking a large number of values, it is very likely for almost every single combination of values to appear fewer times than expected. As a result, our definition of nominal outlyingness considers sequences of nominal values, weighing their contribution to outlyingness in a way that is inversely proportional to the sequence length. This is done in order to ensure that, if interest lies in highly infrequent sequences, combinations of values are seen with greater suspicion if they are shorter in length, as they are typically expected to be more abundant than their supersets in a data set. Based on these considerations, we can now define nominal outlyingness as follows:}

\begin{definition}\label{def:nominalout}
    Nominal outlyingness is the degree to which a sequence of nominal values and its subsets are markedly different from an assumed discrete distribution that has generated them.
\end{definition}
In order to introduce a framework for quantifying nominal outlyingness as given in Definition \ref{def:nominalout}, we borrow some ideas from the association rule mining literature. We briefly outline some basic concepts in the following subsection.

\subsection{Association rule mining background}

Association rule mining is a method used to discover relationships of interest between sets of data items, with the seminal work of \cite{agrawal1993mining} having been applied in fields like market basket analysis \citep{unvan2021market} and customer segmentation \citep{silva2019association}. 
%\textcolor{red}{The work of \cite{agrawal1993mining} was one of the first to focus on generating association rules between items in a large database of customer transactions and has paved the way for related concepts to be used in market basket analysis \citep{unvan2021market}, customer segmentation \citep{silva2019association} and other application areas.} %\textcolor{red}{Below,} 
We present some basic notions used in association rule mining. %\textcolor{red}{We later make use of these concepts in the development of a suitable score that quantifies how outlying a nominal observation is.}

%\textcolor{red}{We start by defining}
An \textit{itemset} is a set of items which occur together. Here, \textit{item} refers to the value any variable takes for each observation
%\textcolor{red}{. Since we are working with nominal data, an itemset is simply a sequence of categorical levels.} 
, that is, a sequence of categorical levels in our setting. We denote an itemset by $d$ and its length by $\lvert d\rvert$. Notice that $1 \leq \lvert d \rvert \leq p$, where $p$ is the number of variables and $\lvert d \rvert = 1$ refers to one level of a single nominal variable. Finally, we define the \textit{support} of a an itemset $d$, denoted by $\operatorname{supp}(d)$, as the number of times the itemset $d$ appears in the data.

Given a large data set, some itemsets may appear more often than others. 
%\textcolor{red}{In an outlier identification problem, where the focus is on highly frequent or highly infrequent occurrences, determining whether an itemset is frequent or infrequent is of a major importance.}
This motivates the definition of a notion of frequency. An itemset can be considered as \textit{frequent} when its support exceeds a \textit{minimum support threshold} value, typically denoted by $\sigma$; otherwise it is \textit{infrequent}. 
%\textcolor{red}{As we will see,} 
This threshold value may differ according to the itemset that is being looked at, therefore we use the notation $\sigma_d$ instead, to denote that the threshold value is specific to an itemset $d$. Determining $\sigma_d$ is one of the most challenging problems in the association rule mining literature; this is something that we deal with later in the paper. 

Finally, we introduce the \textit{Apriori principle} and the related concept of \textit{support-based pruning}. According to the Apriori principle \citep{agrawal1994fast}, if an itemset is frequent, then all of its subsets must be frequent as well. Equivalently, if an itemset is infrequent, then all its supersets are also infrequent. A direct implication of the Apriori principle is that one may not consider subsets or supersets of itemsets that are already found to be frequent or infrequent, respectively, something known as support-based pruning. The reason why support-based pruning is used is to avoid redundant computations and hence reduce the computational complexity of the search by trimming the space of itemsets that are being examined.

\section{Proposed framework}\label{sec:proposal}

The problem of outlier identification is generally complex in the sense that it is an unsupervised problem and thus, we can not be certain of whether an observation is truly outlying. However, what we can do is measure outlyingness with respect to a given scheme. By doing so, not only can we deduce how likely an observation is to be an outlier, but we can also compare observations in terms of their outlyingness. In this Section, we introduce a score that quantifies nominal outlyingness given Definition \ref{def:nominalout}. 
%\textcolor{red}{We explain the rationale behind our score formulation and describe how values for the involved hyperparameters can be chosen. We finally introduce the concepts of \textit{nominal outlyingness depth} and the \textit{contribution matrix} which can be used for meta-analysis purposes.}

\subsection{Score of nominal outlyingness}

Let us assume that we are working with a dataset $\mathcal{D}$ consisting of $p$ nominal variables $\boldsymbol{X}_1, \ldots, \boldsymbol{X}_p$. We denote the number of levels for each variable $\boldsymbol{X}_i$ by $\ell_i$ and assume without loss of generality that its levels are encoded by the natural numbers $1, \ldots, \ell_i$ $(i = 1, \ldots, p)$. 
%\textcolor{red}{Having defined some basic notions from the association rule mining literature,} 
We treat each entry in a contingency table as an individual itemset; the total number of contingency tables is given by $2^p - 1$ and the total number of itemsets that one can encounter is $\prod_{i=1}^{p}\left(\ell_i+1 \right)-1$.
%\textcolor{red}{, assuming that an individual nominal feature $\boldsymbol{X}_i$ is a $(\ell_i \times 1)$-dimensional contingency table.} 
We further define $\boldsymbol{\pi}_i$ to be the $\ell_i$-dimensional vector including the expected proportions for each level of $\boldsymbol{X}_i$ $(i = 1, \ldots, p )$. 

We define a \textit{score of nominal outlyingness} in a similar manner to \cite{otey2006fast}.
%\textcolor{red}{, who use ideas from the frequent itemset mining literature to introduce links between variables in the categorical space.} 
Given $n$ observations $\boldsymbol{x}_1, \ldots, \boldsymbol{x}_n \in \mathcal{A}$, where $\mathcal{A} = \bigtimes_{i=1}^p \{1, \ldots, \ell_i\}\subset \mathbb{N}^p$ and $\bigtimes$ denotes the Cartesian product, we define the score of nominal outlyingness for an observation $\boldsymbol{x}_i$ as
\begin{align}\label{eq:discretescoreinfrequent}
    %s(\boldsymbol{x}_i)=\sum_{\substack{d \subseteq \boldsymbol{x}_{i}: \\ \operatorname{supp}(d) \notin (\sigma_d, n), \\ \lvert d \rvert \leq \maxlen
    %, \\ \left\{\left\{k, k'\right\}: u_{k,k'}>u^{\mathrm{upper}}_{k,k'}\right\}\nsubseteq d
    %}}\frac{\log(\alpha) - \log \mathbb{P}(\boldsymbol{X} = \boldsymbol{x}_i \mid d)}{\prod\limits_{j \subseteq d} \ell_j \times \lvert d \rvert^r}, \quad  \alpha, r> 0, \ i=1,\dots,n.
    s(\boldsymbol{x}_i)=\sum_{\substack{d \subseteq \boldsymbol{x}_{i}: \\ \operatorname{supp}(d) \notin (\sigma_d, n], \\ \lvert d \rvert \leq \maxlen}} \frac{\sigma_d}{\operatorname{supp}(d) \times \lvert d \rvert^r}, \\
    r> 0, \ i=1,\dots,n, \nonumber
\end{align}
if we are interested in highly infrequent itemsets. 
%\textcolor{red}{Recall that $\operatorname{supp}(d)$ refers to the support of itemset $d$, as defined in Section \ref{sec:defoutliers}. Otherwise,} 
For highly frequent itemsets, the score of nominal outlyingness becomes
\begin{align}\label{eq:discretescorefrequent}
    %s(\boldsymbol{x}_i)=\sum_{\substack{d \subseteq \boldsymbol{x}_{i}: \\ \operatorname{supp}(d) \notin (0, \sigma_d), \\ \lvert d \rvert \leq \maxlen
    %, \\ \left\{\left\{k, k'\right\}: u_{k,k'}>u^{\mathrm{upper}}_{k,k'}\right\}\nsubseteq d
    %}}\frac{\log(\alpha) - \log \mathbb{P}(\boldsymbol{X} = \boldsymbol{x}_i \mid d)}{\prod\limits_{j \subseteq d} \ell_j \times \lvert d \rvert^r}, \quad  \alpha, r> 0, \ i=1,\dots,n.
    s(\boldsymbol{x}_i)=\sum_{\substack{d \subseteq \boldsymbol{x}_{i}: \\ \operatorname{supp}(d) \notin [0, \sigma_d), \\ \lvert d \rvert \leq \maxlen}} \frac{\operatorname{supp}(d)}{\sigma_d \times \left( \maxlen - \lvert d \rvert + 1 \right)^r}, \\
    r > 0, \ i=1,\dots,n. \nonumber
\end{align}

For the rest of this section, we assume that nominal outliers are highly infrequent observations with respect to an underlying discrete distribution. Hence, our arguments involve Expression \eqref{eq:discretescoreinfrequent}; extending these to the case of highly frequent itemsets is straightforward.

The formulation in Expression \eqref{eq:discretescoreinfrequent} %\textcolor{red}{begins by looking at itemsets of unit length and compares their support to a lower threshold value $\sigma_d$. If the support lies outside $(\sigma_d, n]$, the score of all points possessing that itemset is augmented in a way that is directly proportional to $\sigma_d$ and inversely proportional to the product of the itemset support and the itemset length raised to a non-negative power $r$.} 
augments the score for observations possessing items of length up to $\maxlen$ \citep[this is alternative notation to what][call $\text{MAXLEN}$]{koufakou2007scalable} with support outside $(\sigma_d, n]$ proportionally to $\sigma_d/\operatorname{supp}(d)$. This improves upon \cite{koufakou2007scalable} and \cite{koufakou2010} by accounting for relative proximity to infrequency. 
%\textcolor{red}{One can immediately notice the similarity of Expression \eqref{eq:discretescoreinfrequent} (with $r=1$ and a numerator of a unit) to the one used by \cite{koufakou2007scalable} and \cite{koufakou2010}. The idea behind modifying the score to compute $\sigma_d/\operatorname{supp}(d)$ instead of just $1/\operatorname{supp}(d)$ is that we wish to ensure a more fair contribution of itemsets with respect to their expected probabilities.} 
%\textcolor{red}{For instance, if two highly infrequent itemsets $d_1, d_2$ of unit occur equally often in a data set, the score of \cite{koufakou2010} would treat them as equally outlying. Yet, if the associated probabilities satisfy $\pi_{d_1} < \pi_{d_2}$, the corresponding threshold values are expected to obey $\sigma_{d_1} < \sigma_{d_2}$. In this case, it is more natural to view $d_2$ as more outlying than $d_1$, since the former was expected to occur much more frequently than the latter.} 
For example, two itemsets may occur equally often, but if one is expected to be much more common, the ratio $\sigma_d/\operatorname{supp}(d)$ assigns it a higher degree of outlyingness. 

Division by the itemset length ensures that naturally rarer, longer itemsets contribute less to the score than shorter ones, while the exponent $r$ tunes how sharply this length penalty is applied. In practice, this enables unusual single attributes (e.g. a rare occupation or diagnosis code) to provide clearer evidence of outlyingness than a rare combination of many attributes, which may simply be a consequence of combinatorial sparsity. 
%\textcolor{red}{Computing the ratio $\sigma_d/\operatorname{supp}(d)$ is a way to account for this and ensure that the score is augmented proportionally to the ratio of the minimum expected over the actual support of each highly infrequent itemset.}

Despite not being explicitly mentioned in Expression \eqref{eq:discretescoreinfrequent}, support-based pruning is implemented to exclude supersets of any highly infrequent itemsets from subsequent computations. That is, if $\lvert d_1 \rvert < \maxlen$ and $\operatorname{supp}(d_1) < \sigma_{d_1}$, then any itemset $d_2 \supset d_1$ is not checked in subsequent calculations. 
%\textcolor{red}{Support-based pruning, as also implemented in \cite{koufakou2010}, is very convenient when the score is computed in an ascending order from $\lvert d \rvert = 1$ up to $\maxlen$. However, this is only useful in the case of looking for highly infrequent itemsets. In contrast, using the Apriori principle per se is much more useful when we start looking at itemsets of length $\maxlen$ until we reach individual categorical levels; this approach is favoured when highly frequent itemsets are of interest instead.} 
We illustrate support-based pruning for a data point $\boldsymbol{x}$ consisting of three nominal levels, say $\boldsymbol{x} = (1, 1, 1)^\intercal$, in Figure \ref{fig:treesearchexample}. We use a tree representation to illustrate how the search for highly frequent or infrequent itemsets is being conducted.

\begin{figure}
    \begin{subfigure}[h!]{0.45\linewidth}
        \includegraphics[width=\linewidth]{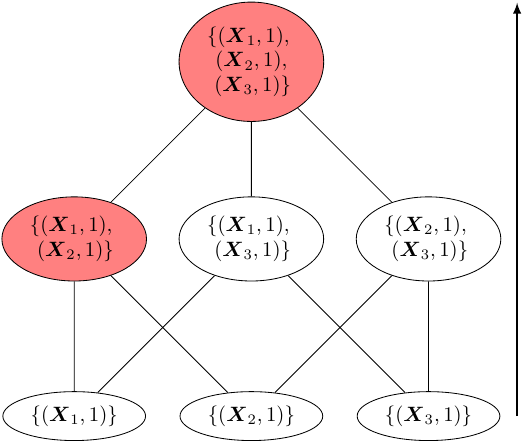}
        \caption{Bottom-up search for highly infrequent itemsets.}
        \label{subfig:bottomupsearch}
    \end{subfigure}
    \hfill
    \begin{subfigure}[h!]{0.45\linewidth}
        \includegraphics[width=\linewidth]{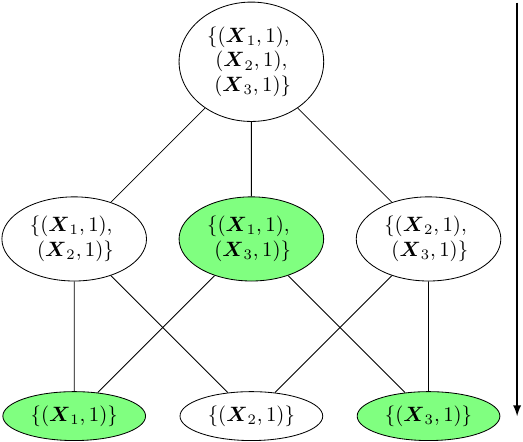}
        \caption{Top-down search for highly frequent itemsets.}
        \label{subfig:topdownsearch}
    \end{subfigure}%
    \caption{Implementation of the Apriori principle in support-based pruning for a toy example using a tree representation. Notation $(\boldsymbol{X}_i, l)$ means $\boldsymbol{X}_i = l$. Red and green nodes represent highly infrequent and highly frequent itemsets, respectively.}
    \label{fig:treesearchexample}
\end{figure}

%\textcolor{red}{Figure \ref{fig:treesearchexample} illustrates how the algorithm proceeds if the user is interested in highly infrequent or highly frequent itemsets.} 
A bottom-up search is used when infrequent itemsets are of interest; this bottom-up traversal of the tree allows for highly infrequent itemsets to be detected and any of their supersets to be removed from subsequent searching. For instance, in Figure \ref{subfig:bottomupsearch} the itemset consisting of $\boldsymbol{X}_1 = 1$ and $\boldsymbol{X}_2 = 1$ is infrequent, thus the itemset that further includes the value of $\boldsymbol{X}_3$ is not considered in the subsequent search. If highly frequent itemsets are of interest instead, the method proceeds analogously using a top-down approach instead, as shown in Figure \ref{subfig:topdownsearch}. Therefore, assuming that $\operatorname{supp}((\boldsymbol{X}_1, 1), (\boldsymbol{X}_2, 1)) = v$, $\maxlen > 2$, and the corresponding minimum support threshold value is $u > v$, we get $s(\boldsymbol{x}) = u/(2^rv)$; this is the same for any observation satisfying $\boldsymbol{X}_1=1$ and $\boldsymbol{X}_2 = 1$.
%\textcolor{red}{In this case, if the itemset consisting of $\boldsymbol{X}_1 = 1$ and $\boldsymbol{X}_3 = 1$ is flagged as frequent, all its subsets will also be marked as such and will therefore not be included in the rest of the search.}

%\textcolor{red}{Our main contributions include a method for determining the support threshold values $\sigma_d$ that are specific to each itemset $d$. Moreover, we propose a suitable stopping criterion, that is a $\maxlen$ value, that accounts for the sparsity introduced in higher-dimensional contingency tables. We also derive some mathematical properties of $s(\cdot)$, such as when $s(\cdot)$ is maximised and what its maximum value is; these are included in Propositions \ref{prop:maxdiscscore_expression} \& \ref{prop:maxdiscscore} in Section \ref{sec:properties}.}

\subsection{Choice of hyperparameter values}

The hyperparameters involved in the calculation of the score $s(\cdot)$ are the minimum support threshold $\sigma_d$, the maximum length of sequences considered $\maxlen$, and the exponent term $r$. We start by looking at the minimum support threshold value $\sigma_d$. This is defined with respect to the assumed underlying Multinomial distribution for each itemset. Assume without loss of generality that we are working with itemsets $d$ consisting of the first $\lvert d \rvert$ nominal variables. For such itemsets, denoted as $\boldsymbol{X}_{[1: \lvert d \rvert]}$, the underlying Multinomial distribution that we consider is
\begin{equation*}
    \boldsymbol{X}_{[1: \lvert d \rvert]} = \boldsymbol{x} \mid d \sim \text{Multinomial}\left( 1, \boldsymbol{p}^d \defeq \bigotimes\limits_{i=1}^{\lvert d \rvert} \boldsymbol{\pi}_i \right),
\end{equation*}
where $\otimes$ is the Kronecker product. Notice how this formulation implicitly assumes independence of the nominal variables; a dependence structure between the nominal variables does not change anything besides the definition of the probability vector, which has to be explicitly provided by the user instead. 

%\textcolor{red}{Returning to our original problem of determining a suitable threshold value for the support of an itemset $d$, we compute the support of the itemset and see if it falls within a specific range of reasonable values under the assumed model.} 
In order to check if the support of an itemset is within a sensible range of values under the assumed model, we construct an upper one-sided $100(1-\alpha)\%$ confidence interval for the probability of observing itemset $d$. Since all itemsets of length $\lvert d \rvert$ form a Multinomial distribution with $\prod_{i=1}^{\lvert d \rvert} \ell_i$ possible outcomes, we use the method of \cite{sison1995simultaneous} to construct $100(1-\alpha)\%$ simultaneous confidence intervals for the Multinomial proportions. 
%\textcolor{red}{The confidence intervals are constructed by approximating the Multinomial cumulative distribution function (CDF) using a theorem from \cite{levin1981representation}. The theorem approximates the Multinomial CDF as the product of Poisson probabilities multiplied by the probability mass function of a sum of truncated Poisson random variables, which is approximated using Edgeworth series.} 
The two-sided $100(1-\alpha)\%$ confidence intervals for the components of the Multinomial proportions vector $\boldsymbol{p}^d$ are given by
\begin{equation}\label{eq:sisonglazci}
    \left(p^d_i - \frac{c_\alpha}{n}, p^d_i + \frac{c_\alpha+2\gamma_\alpha}{n}\right), \quad i = 1, \ldots, \prod_{j=1}^{\lvert d \rvert}\ell_j,
\end{equation}
where $\gamma_\alpha = (1-\alpha-\nu(c_\alpha))/(\nu(c_\alpha+1)-\nu(c_\alpha))$, $\nu(c_\alpha) = \mathbb{P}(x_i - c_\alpha \leq X_i \leq x_i + c_\alpha)$ with $X_i \sim \text{Multinomial}(n, p^d_i)$ and $c_\alpha \in \mathbb{Z}$ satisfies $\nu(c_\alpha) < 1-\alpha < \nu(c_\alpha+1)$.

In a similar manner, the upper and lower one-sided $100(1-\alpha)\%$ confidence intervals are given by 
\begin{align*}
    \left(p^d_i - \frac{c_{2\alpha}}{n}, 1\right), & \quad i = 1, \ldots, \prod_{j=1}^{\lvert d \rvert}\ell_j, \\
    \left(0, p^d_i + \frac{c_{2\alpha}+2\gamma_{2\alpha}}{n}\right), & \quad i = 1, \ldots, \prod_{j=1}^{\lvert d \rvert}\ell_j.
\end{align*}
%\textcolor{red}{The confidence intervals obtained for the Multinomial proportions are used to determine what the least expected values for the itemset supports should be.} 
Multiplying the lower bound of each upper one-sided interval by the sample size $n$ yields bounding values for the support of each itemset which are used as our values for $\sigma_d$; these are no longer universal but specific to each itemset and its assumed frequency.
%\textcolor{red}{, determining when an itemset is highly infrequent with respect to the underlying Multinomial distribution.} 
%\textcolor{red}{This extends the idea of \cite{otey2006fast} and \cite{koufakou2010}, who also augmented the nominal outlyingness scores of observations possessing highly infrequent itemsets using an expression similar to \eqref{eq:discretescoreinfrequent}, but with a constant $\sigma$ value that was the same for all itemsets of any length.}

As the itemset length grows, the associated contingency tables become more sparse. This leads to a serious issue if one assumes that $p_i^d \overset{\lvert d \rvert \rightarrow p}{\longrightarrow} 0$, since almost any itemset is considered as highly infrequent and there is a substantial and potentially prohibitive increase in computational resources required. 
%\textcolor{red}{As a result, we seek to find a maximum allowed value for $\lvert d \rvert$, which we denote by $\maxlen$, so that itemsets of length up to $\maxlen$ will be included in our score computations.} 
Determining a stopping criterion $\maxlen$ so that it reflects on our concerns regarding sparsity can be done by requiring all upper thresholds $\sigma_d$ to be at least equal to two for all itemsets $d$. The rationale behind this is that if $\sigma_d < 2 \ \forall d : \lvert d \rvert = M$, then only itemsets of length $M \leq p$ with unit support are not flagged as infrequent. This is a restrictive case leading to most (if not all) observations having their scores augmented due to being more frequent than expected. Thus, we define $\maxlen$ as follows
\begin{align*}
    \maxlen = \max\left\{ \right.& 1 \leq M \leq p: \sigma_d \geq 2 \\ 
    & \left. \mathrm{for \ all} \ d \subseteq\boldsymbol{x}_i \ \mathrm{with} \ \lvert d \rvert \leq M \right\}.
\end{align*}

%\textcolor{red}{The use of concepts from the association rule mining literature for the proposed score of nominal outlyingness presents a practical issue related to computational complexity. For a given value of $\maxlen$, the total number of contingency tables that can be generated is given by $\sum_{k=1}^{\maxlen}C^p_k$, where $C$ denotes the binomial coefficient. Moreover, the number of possible itemsets also gets very large which shows that under the presence of many variables, the total number of comparisons that need to be made can grow very rapidly, especially for large $\maxlen$. As a result, any computations that may be redundant should be omitted.} 

%\textcolor{red}{Our proposal for the score includes an additional hyperparameter, that is the exponent term $r > 0$.} 
The exponent term $r > 0$ controls the influence of longer itemsets in the nominal outlyingness score. Based on the principle of \cite{koufakou2010} which states that an observation with frequent itemsets of unit length and $k \leq p$ infrequent itemsets of length two should not be seen as suspicious to be an outlier as an observation with $k$ infrequent itemsets of unit length, our framework allows the user to decide the score contribution of longer itemsets. 
%\textcolor{red}{A larger $r$ value will yield a smaller augmentation of $s(\cdot)$ as $\lvert d \rvert$ increases, hence providing additional flexibility and control of how the score of outlyingness is being computed.} 
The choice of $r$ is context-specific and may be selected based on the level of sparsity in the data set. In sparse datasets, where many longer itemsets occur rarely by construction, high values of $r$ may overly penalise these itemsets, while in denser data, moderate values allow longer patterns to contribute meaningfully. For practical purposes, we suggest restricting to $r \leq 3$; when $r = 3$, a two-dimensional itemset requires $\sigma_d/\operatorname{supp}(d) > 8$ to match the contribution of a single attribute. Larger values are possible, but we advise against setting $r$ too large, so that longer itemsets can still play a role without dominating the score.
%\textcolor{red}{We cannot recommend a specific value for $r$ as that is typically context-specific and depends on the weight the user wants to assign to highly infrequent (or frequent) itemsets of greater length. For instance, a reasonable value for $r$ would be $r \leq 3$; in case we set $r=3$, we would need $\sigma_d/\operatorname{supp}(d) > 8$ for a two-dimensional itemset $d$ to contribute to $s(\cdot)$ as much as a specific outlying level on its own. While setting $r > 3$ is certainly possible, we suggest that $r$ should not be much larger than the recommended value of $r=3$, so that longer sequences can also contribute to the score, at least to some extent.}

\subsection{Nominal outlyingness depth \& Contribution matrix}

We introduce the concept of \textit{nominal outlyingness depth}, as an 
%\textcolor{red}{extension of our proposed method for quantifying nominal outlyingness. This is an} 
additional metric 
%\textcolor{red}{that we introduce} 
for meta-analysis purposes.
%\textcolor{red}{so that the scores are easier for the user to understand and interpret.} 
As illustrated in Figure \ref{fig:treesearchexample}, the score $s(\cdot)$ is being computed in a hierarchical process and can therefore be represented using a tree-based structure. This has motivated us to come up with a concept inspired from the continuous anomaly detection literature and more precisely from the Isolation Forest algorithm \citep{liu2008isolation}, where outlyingness is assessed by the \textit{average depth}; that is the average number of random splits that it takes for an observation to be \textit{isolated} in the $p$-dimensional Euclidean space. We 
%\textcolor{red}{proceed in a similar manner and} 
introduce a Nominal Outlyingness Depth concept, which we denote by $\operatorname{NOD}(\cdot)$, given by
\begin{align}\label{eq:nod_def}
    \operatorname{NOD}(\boldsymbol{x}_i) =\left.\sum_{\substack{d \subseteq \boldsymbol{x}_{i}: \\ \operatorname{supp}(d) \notin (\sigma_d, n], \\ \lvert d \rvert \leq \maxlen}} \lvert d \rvert \middle/ \sum_{\substack{d \subseteq \boldsymbol{x}_{i}: \\ \operatorname{supp}(d) \notin (\sigma_d, n], \\ \lvert d \rvert \leq \maxlen}}\right. 1, \\
    i=1,\dots,n. \nonumber
\end{align}
%\textcolor{red}{The rationale behind Expression \eqref{eq:nod_def} is that} 
For each observation $\boldsymbol{x}_i$ we compute the average itemset length over all infrequent itemsets included in $\boldsymbol{x}_i$. A similar Expression to \eqref{eq:nod_def} can be formulated if highly frequent itemsets are of interest.
%\textcolor{red}{; the only modification is changing the interval $(\sigma_d, n]$ to $[0, \sigma_d)$ and $\lvert d \rvert$ on the numerator becomes $\maxlen - \lvert d \rvert + 1$, to account for the fact that we use a top-down search instead. One can also consider the definition of the nominal outlyingness depth}
Using a tree representation, NOD is defined as the average traversal depth of highly infrequent or highly frequent itemsets, where the depth is relative to a base level that depends on whether a bottom-up or a top-down search is being used. Such an illustration is given in Figure \ref{fig:depth_example} for the toy example of $\boldsymbol{x} = (1, 1, 1)^\intercal$. For instance, according to Figure \ref{subfig:depth_bottomup}, we are looking for highly infrequent itemsets and thus we use a bottom-up approach. The nominal outlyingness depth in this case is given by $(1+2)/2 = 1.5$. Similarly for the case of highly frequent itemsets,
%\textcolor{red}{where a top-down search is being employed as in Figure \ref{subfig:depth_topdown}, }
$\text{NOD}(\boldsymbol{x})(2+3)/2 = 2.5$. The depth value ranges from a unit up to $\maxlen$, with a lower value corresponding to a larger degree of outlyingness, except for observations with no highly frequent or highly infrequent itemsets; these have a zero depth.

We remark that the proposed NOD parallels ideas of classical statistical depth functions \citep[see][for a review]{zuo2000general} which measure centrality via distances or directional projections. Here, centrality is instead captured by the frequency structure of the underlying discrete distribution. Observations supported mainly by longer infrequent itemsets are considered more outlying, while those involving shorter, common itemsets are viewed as deeper.

\begin{figure}
    \begin{subfigure}[h!]{0.45\linewidth}
        \includegraphics[width=\linewidth]{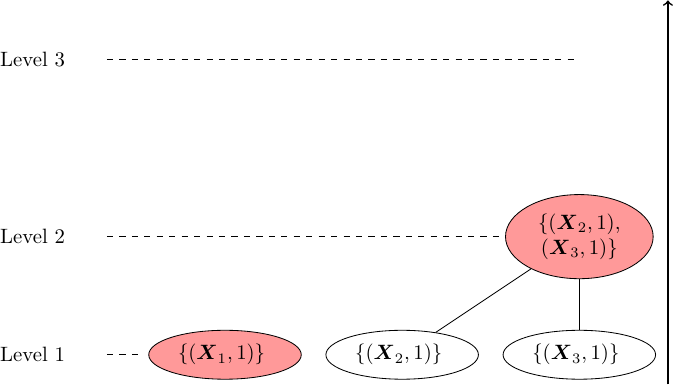}
        \caption{Nominal outlyingness depth calculation for highly infrequent itemsets.}
        \label{subfig:depth_bottomup}
    \end{subfigure}
    \hfill
    \begin{subfigure}[h!]{0.45\linewidth}
        \includegraphics[width=\linewidth]{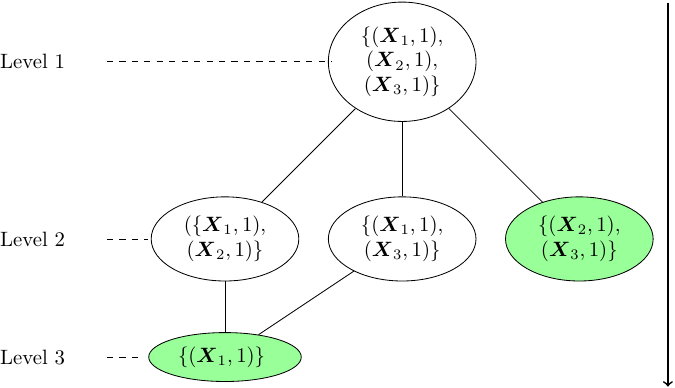}
        \caption{Nominal outlyingness depth calculation for highly frequent itemsets.}
        \label{subfig:depth_topdown}
    \end{subfigure}%
    \caption{Illustration of the calculation of nominal outlyingness depth calculation using a tree representation. Notation $(\boldsymbol{X}_i, l)$ means $\boldsymbol{X}_i = l$. Red and green nodes represent highly infrequent and highly frequent itemsets, respectively. Support-based pruning has been applied in both cases.}
    \label{fig:depth_example}
\end{figure}

The final concept we introduce is that of the \textit{contribution matrix}, which we denote by $\boldsymbol{\mathrm{C}}$. This is a $(n \times p)$-dimensional matrix, with the $(i, j)$th entry being the contribution of variable $j$ to the nominal outlyingness score of $\boldsymbol{x}_i$. For highly infrequent itemsets, this contribution is given by
\begin{align}\label{eq:contrib_infrequent}
    %s(\boldsymbol{x}_i)=\sum_{\substack{d \subseteq \boldsymbol{x}_{i}: \\ \operatorname{supp}(d) \notin (\sigma_d, n), \\ \lvert d \rvert \leq \maxlen
    %, \\ \left\{\left\{k, k'\right\}: u_{k,k'}>u^{\mathrm{upper}}_{k,k'}\right\}\nsubseteq d
    %}}\frac{\log(\alpha) - \log \mathbb{P}(\boldsymbol{X} = \boldsymbol{x}_i \mid d)}{\prod\limits_{j \subseteq d} \ell_j \times \lvert d \rvert^r}, \quad  \alpha, r> 0, \ i=1,\dots,n.
    c_{i,j} = \sum_{\substack{d \subseteq \boldsymbol{x}_{i}: \\ j \subseteq d \\ \operatorname{supp}(d) \notin (\sigma_d, n], \\ \lvert d \rvert \leq \maxlen}} \frac{\sigma_d}{\operatorname{supp}(d) \times \lvert d \rvert^{r+1}}, \quad  r> 0, \\
    i=1,\dots,n, \ j=1, \dots ,p. \nonumber
\end{align}
%\textcolor{red}{As a result,} 
All the elements of the contribution matrix are non-negative and the sum of each row is the score of nominal outlyingness for the relevant observation. 
%\textcolor{red}{that is $\sum_{j=1}^p c_{i,j} =  s(\boldsymbol{x}_i)$.} 
The main idea behind Expression \eqref{eq:contrib_infrequent} is that if an itemset is found to be highly infrequent, all variables involved receive an equal contribution to the augmentation of the score that is caused by this sequence of categorical levels. 
%\textcolor{red}{One can therefore use} 
The contribution matrix can facilitate understanding of the source of nominal outlyingness of the data by seeing which variables are involved in highly infrequent (or frequent) itemsets.
%\textcolor{red}{, causing an increase of the score of nominal outlyingness.}

\subsection{The SONO algorithm}

Based on the previous subsections, we present the pseudo-code for our algorithm in Algorithm \ref{code:sono}, which we call SONO (Scores Of Nominal Outlyingness). The code assumes highly infrequent itemsets are of interest.

\begin{algorithm*}[!ht]
\caption{The SONO Algorithm for highly infrequent itemsets}
\label{code:sono}
\begin{algorithmic}[1]
\State \textbf{Input:} Data set $\boldsymbol{X}$ with $n$ observations, $p$ nominal variables, exponent term $r>0$, significance level $\alpha > 0$, expected proportions of variable levels $(\boldsymbol{\pi}^\intercal_1, \ldots, \boldsymbol{\pi}^\intercal_p)^\intercal$. (Optional: $L_{\max}$ value)
\State \textbf{Output:} Scores of nominal outlyingness $\{s(\boldsymbol{x}_i)\}_{i=1}^n$
\Statex

\State Estimate $L_{\max}$ based on $(\boldsymbol{\pi}^\intercal_1, \ldots, \boldsymbol{\pi}^\intercal_p)^\intercal$ (unless user-provided)
\State Initialise list of infrequent itemsets $\mathcal{I} \gets \varnothing$
\State Set $M \gets 1$
\While{$M \leq L_{\max}$}
  \State Record supports of all itemsets not in $\mathcal{I}$ of length $M$ in $\boldsymbol{X}$
  \State Construct simultaneous $100(1-\alpha)\%$ upper CIs for each itemset of length $|d|$ and determine $\sigma_d$ values
  \If{$\exists d \notin \mathcal{I}$ such that $\operatorname{supp}(d) < \sigma_d$ and $|d| = M$}
    \State Augment $s(\boldsymbol{x}_i)$ for all $\boldsymbol{x}_i \supseteq d$
    \State Flag $d$ and its supersets as infrequent; $\mathcal{I} \gets \mathcal{I} \cup \{d' : d' \supseteq d\}$
  \EndIf
\EndWhile
\end{algorithmic}
\end{algorithm*}

\section{Properties of the proposed score}\label{sec:properties}

We now explore some properties of our proposed score of nominal outlyingness. Notice that we focus exclusively on the case of an outlier being defined as consisting of highly infrequent itemsets; extending these properties to the frequent case is straightforward. We are mainly interested in the range of values the score can take, as well as on when it is maximised; essentially, we would like to know under what circumstances an observation is extremely likely to be outlying and how the associated hyperparameters influence its score. Due to the complexity of the problem, the following statements rely on certain assumptions that represent rather extreme cases but they still provide useful intuition for interpreting the scores.

\begin{proposition}\label{prop:maxdiscscore_expression}
    Let $\maxlen = p$ and assume the minimum support threshold value is given by its maximum possible value so that $\sigma_d/\operatorname{supp}(d) = n-\delta$, where $\delta \rightarrow 1$ for all itemsets $d$. The maximum value of the nominal score of outlyingness $s(\boldsymbol{x}_i)$ for infrequent itemsets in a data set with $p \geq 2$ nominal variables is attained for all itemsets of length $1 \leq k \leq p$ appearing just once.
\end{proposition}
\begin{proof}
    %\textcolor{red}{We assume that we have $p \geq 2$ nominal variables (for the case of $p=1$, the result follows immediately by the definition of the nominal score) and recall the formulation of the discrete score for an observation $1 \leq i \leq n$ as given in Expression \eqref{eq:discretescoreinfrequent}$$     s(\boldsymbol{x}_i)=\sum_{\substack{d \subseteq \boldsymbol{x}_{i}: \\ \operatorname{supp}(d) \notin (\sigma_d, n], \\ \lvert d \rvert \leq \maxlen}} \frac{\sigma_d}{\operatorname{supp}(d) \times \lvert d \rvert^r}, \quad  r> 0.$$ Assuming a large $\maxlen$ value (to avoid having many restrictions on the calculation of the score), it is easy to see that itemsets of unit support with a large minimum support threshold value $\sigma_d$ yield a higher increase in the score. This case corresponds to highly misspecified large probabilities for itemsets that only appear once in the data, for which $\sigma_d = n - \delta$, where $\delta$ is typically a small positive integer. Since $\sigma_d$ decreases as $\lvert d \rvert$ gets larger and as this is a parameter that is hard to determine in advance, we will stick with the maximisation of the $1/\lvert d \rvert ^r$ term, assuming the maximum possible $\sigma_d$ value of $\sigma_d = n-1$.}
    Recalling the score in Expression \eqref{eq:discretescoreinfrequent} and given these assumptions, the total score is equivalently written as
    \begin{multline*}
    \mathcal{A} = (n-1) \times \\
     \times \left\{ p - \sum\limits_{i=1}^{p-1}\alpha_i + \sum\limits_{i=1}^{p-1}\dbinom{\alpha_i}{i+1}\frac{1}{\left(i+1\right)^r}\right\}, \\
    \alpha_i \in \left\{0, i+1, \dots, p\right\}, \ \sum\limits_{i=1}^{p-1}\alpha_i \leq p.
    \end{multline*}
    %\textcolor{red}{Notice that the first two terms inside the curly brackets correspond to the contribution of the itemsets of unit length that appear once in the data set, while the third term corresponds to the contribution of all possible itemsets of greater length (up to length $p$) to the score.}
    The coefficients $\alpha_i$ represent the number of nominal variables which are included in infrequent itemsets (of unit support) of length $i+1$. We further impose the restriction that the sum of the $\alpha_i$'s is at most equal to $p$, since any itemsets of length $i+1$ are defined based on $i+1$ variables, and we are restricted to $p$ nominal features. Moreover, we may have no infrequent itemsets of length $i+1$ (in which case $\alpha_i=0$) but if we do have any, then these should be observed in at least $i+1$ nominal features. Despite the abuse of notation, we assume that $C^0_{i+1}=0$ for convenience, with $C$ denoting the binomial coefficient.\\
    
    %\textcolor{red}{Our goal is to find the set of values $\boldsymbol{\alpha} =\left(\alpha_1, \dots, \alpha_{p-1} \right)$ that maximise $\mathcal{A}$ and more precisely, we aim to show that this is achieved either when all the $\alpha_i$'s are equal to zero or when we are on the boundary of the solution space, i.e. when all the $\alpha_i$'s sum to $p$ and all of them but one are exactly zero. The first case corresponds to maximising the nominal score for all $p$ nominal features of an observation being unique within the data set, while the second is interpreted as achieving the maximum nominal score possible when all itemsets of length $i+1$ that can be generated from $p$ nominal variables appear just once and there exist no itemsets of greater or lower length which are infrequent. In both these cases, we basically end up with the conclusion that the maximum score is attained for all itemsets of one specific length occurring once in the data set.\\}
    
    For $p < 2^{r+1} + 1$, we would rather have $\alpha_i=0 \ \forall i$ if
    $$
    \sum\limits_{i=1}^{p-1}\alpha_i > \sum\limits_{i=1}^{p-1}\dbinom{\alpha_i}{i+1}\frac{1}{(i+1)^r}.$$
    Equivalently
    $$
    \sum\limits_{i=1}^{p-1}\left\{ \alpha_i \times \left( \frac{\left(\alpha_i-1\right) \times \dots \times \left(\alpha_i-i\right)}{(i+1)^{r+1}\times i!}-1\right)\right\} < 0.$$
    Since we assume that all the $\alpha_i$'s are non-negative, the expression above can be negative if and only if $\exists \alpha_i > 0$ such that
    \begin{align}\label{eq:aiszero}
        F & = \frac{\left(\alpha_i-1\right) \times \dots \times \left(\alpha_i-i\right)}{(i+1)^{r+1}\times i!}-1 \\
        & = \frac{\Gamma\left(\alpha_i\right)}{\Gamma\left(\alpha_i-i\right)\times\Gamma\left(i+1\right)\times (i+1)^{r+1}} - 1 < 0, \nonumber
    \end{align}
    where $\Gamma(\cdot)$ is the gamma function. If $F<0\ \forall \alpha_i > 0$, we can ensure that $\mathcal{A}$ is maximised for all the $\alpha_i$'s being equal to zero. Since Expression \eqref{eq:aiszero} is maximised when $i=1$ and $\alpha_1 = p$, we obtain
    \begin{align*}
        F = \frac{p-1}{2^{r+1}} < 1 \iff p < 2^{r+1} + 1.
    \end{align*}
        Therefore, $\alpha_i = 0 \ \forall i$ maximises expression $\mathcal{A}$ for $p < 2^{r+1}+1$.\\

    We now assume $p \geq 2^{r+1}+1$. 
    %\textcolor{red}{We first show that getting to the solution boundary is always preferable when having one non-zero $\alpha_i$ and then, we show that setting $\alpha_j$ equal to a non-zero value, where $i \neq j$, decreases the value of $\mathcal{A}$.} 
    Assuming that there exists an index $j$ such that $\alpha_j > 0$ and $\alpha_i=0 \ \forall i \neq j$, we look at what happens to the value of $\mathcal{A}$ if we decrease $\alpha_j$ by a unit. Due to the restriction on the values that $\alpha_j$ can take, we need to assume that $j \neq p -1$, since $\alpha_{p-1}$ can only be equal to 0 or $p$. We denote the loss that results from decreasing $j$ by a unit by $G_1$ and that is equal to
\begin{align*}
    G_1 & = p - \alpha_j + \dbinom{\alpha_j}{j+1}\times \frac{1}{(j+1)^r} -\\
    & - \left\{ p - \alpha_j + 1 + \dbinom{\alpha_j-1}{j+1}\times \frac{1}{(j+1)^r}\right\}\\
    & = \frac{1}{(j+1)^r}\left\{ \dbinom{\alpha_j}{j+1} - \dbinom{\alpha_j-1}{j+1} \right\}-1\\
    & = \frac{\Gamma\left(\alpha_j\right)}{(j+1)^{r-1}\times \Gamma\left(j+2\right)\times \Gamma\left( \alpha_j-j\right)} - 1.
\end{align*}
When decreasing $\alpha_j$ from its minimum possible non-zero value of $j+1$ to zero, we define $G_2$ to be the loss corresponding to this decrease
\begin{align*}
    G_2 & = p - \alpha_j + \dbinom{\alpha_j}{j+1}\times \frac{1}{(j+1)^r} - p\\
    & = -(j+1) + \dbinom{j+1}{j+1}\times\frac{1}{(j+1)^r}\\
    & = \frac{1}{(j+1)^r}-(j+1).
\end{align*}
Hence, $G_2 <0 \ \forall j \in \mathbb{Z}^+$%\textcolor{red}{, which implies $\mathcal{A}$ increases when going from $\alpha_j=j+1$ to $\alpha_j=0$.} 
Recalling that $\alpha_j > j$, we may write $\alpha_j = j + z$, where $z \in \{ 1, \ldots, p-j\}$. The expression for $G_1$ then becomes
\begin{align*}
    G_1 = \frac{\Gamma\left(j+z\right)}{(j+1)^{r-1}\times \Gamma\left(j+2\right)\times \Gamma\left( z\right)} - 1.
\end{align*}
The above expression is increasing in $z$ 
%\textcolor{red}{(this can be seen by considering the ratio $\Gamma(j+z)/\Gamma(z)$ and using Stirling's approximation for the gamma function)} 
and decreasing in $j$. Thus $G_1$ is minimised when $j=p-2$ and $z=1$, in which case
\begin{align*}
    G_1 & = \frac{\Gamma\left(p-1\right)}{(p-1)^{r-1}\times \Gamma\left(p\right)\times \Gamma\left( 1\right)} - 1\\
    & = \frac{1}{(p-1)^r}-1 < 0,
\end{align*}
since $p \geq 2^{r+1} + 1 > 2$ (recall $r>0$). Thus, when we are on the boundary of the solution space, it is preferred for us to stay there, while if we are at the minimum non-zero value of $\alpha_j$, it is preferred to simply set it to zero.\\

The final part of the proof consists of showing that if we are on the boundary of the solution space, with $\alpha_j = p$ and $\alpha_i=0 \ \forall i \neq j$, then activating $\alpha_k$ so that $\alpha_k = k+1$ (where $k \neq j$) leads to a decrease in the value of $\mathcal{A}$. Now clearly, if $\alpha_k$ is set equal to $k+1$, the value of $\alpha_j$ has to drop to $p-(k+1)$, so that the boundary constraint is not violated. We define the loss incurred by activating $\alpha_k$ by $H$ as
\begin{align*}
    H & = \dbinom{p}{j+1}\times \frac{1}{(j+1)^r} - \\ 
    & - \left\{ \dbinom{p-k-1}{j+1}\frac{1}{(j+1)^r} + \dbinom{k+1}{k+1}\frac{1}{(k+1)^r} \right\}\\
    & = \frac{1}{(j+1)^r}\times \left\{\dbinom{p}{j+1} - \dbinom{p-k-1}{j+1} \right\} - \\
    & - \frac{1}{(k+1)^r}\\ 
    & = \frac{1}{(j+1)^r \times\Gamma\left(j+2\right)} \times \\
    & \times \left\{ \frac{\Gamma\left(p+1\right)}{\Gamma\left(p-j\right)} - \frac{\Gamma\left(p-k\right)}{\Gamma\left(p-k-j-1\right)}\right\} - \frac{1}{(k+1)^r}.
\end{align*}
%\textcolor{red}{We distinguish between two possible cases, namely $j < k$ and $j >k$.} 
In the case $j < k$, we obtain
\begin{align*}
    H & = \frac{1}{(k-\beta+1)^r \times\Gamma\left(k-\beta+2\right)}\times \\
    & \times \left\{ \frac{\Gamma\left(p+1\right)}{\Gamma\left(p-k+\beta\right)} - \frac{\Gamma\left(p-k\right)}{\Gamma\left(p-2k+\beta-1\right)}\right\} -\\
    & - \frac{1}{(k+1)^r},
\end{align*}
while if $j > k$ we get
\begin{align*}
    H & = \frac{1}{(k+\beta+1)^r \times\Gamma\left(k+\beta+2\right)}\times\\
    & \times \left\{ \frac{\Gamma\left(p+1\right)}{\Gamma\left(p-k-\beta\right)} - \frac{\Gamma\left(p-k\right)}{\Gamma\left(p-2k-\beta-1\right)}\right\} - \\
    & - \frac{1}{(k+1)^r},
\end{align*}
where $\beta \in \mathbb{Z}^+$. These expressions for $H$ give us additional restrictions, which are that $p-2k+\beta-1 > 0$ and $p-2k-\beta-1 > 0$, so that all terms are well-defined. We now show that $H>0 \ \forall k \neq j$ for $j<k$; the second case follows analogously. Let the first term of $H$ be denoted by $H_1$ and denote the second one by $H_2$. It is obvious that $\beta = 1$ minimises $H_1$. Thus, we have
\begin{align*}
    H & = \frac{1}{k^r \times \Gamma(k+1)} \left\{\frac{\Gamma(p+1)}{\Gamma(p-k)} - \frac{\Gamma(p-k)}{\Gamma(p-2k)} \right\} - \\
    & - \frac{1}{(k+1)^r}.
\end{align*}
The ratio of $H_1$ over $H_2$ is
\begin{align*}
    \frac{H_1}{H_2} = \frac{(k+1)^r \times Q(k)}{k^r \times \Gamma(k+1)} = \frac{(1+1/k)^r \times Q(k)}{\Gamma(k+1)},
\end{align*}
where $Q(k) = \Gamma(p+1)/\Gamma(p-k) - \Gamma(p-k)/\Gamma(p-2k)$. It is also easy to see that $Q(k)$ is increasing for $k$, therefore its minimum is attained when $k=2$. In that case, we have
\begin{align*}
    \frac{H_1}{H_2} & = \frac{(k+1)^r \times Q(k)}{k^r \times \Gamma(k+1)}\\
    & = \frac{(1+1/k)^r \times Q(k)}{\Gamma(k+1)}\\
    & \geq \frac{(3/2)^r \times Q(2)}{2} > 1.
\end{align*}
Thus, $H>0 \ \forall k \neq j$. As a result activating $\alpha_k$ is not preferred as it leads to a positive loss value.\\

Thus, we have shown the following:
\begin{enumerate}
    \item The greatest $\mathcal{A}$ value for $p < 2^{r+1}+1$ is achieved when all the $\alpha_i$'s are equal to zero.
    \item When $p \geq 2^{r+1}+1$, $\mathcal{A}$ is maximised when we are on the boundary of the solution space.
    \item If $p \geq 2^{r+1}+1$ and we are on the boundary of the solution space with just one $\alpha_j$ being equal to $p$, then activating any other parameter $\alpha_k$ (where $j\neq k$) so that it is not longer equal to zero, leads to a lower value of $\mathcal{A}$.
\end{enumerate}
Given the above, we can conclude that the optimal solution for $p < 2^{r+1}+1$ is attained when all parameters are equal to zero, while for $p \geq 2^{r+1}+1$, the optimal solution is on the boundary of the solution space with just one parameter $\alpha_i$ being non-zero and thus equal to $p$, as required.
\end{proof}

Given Proposition \ref{prop:maxdiscscore_expression}, we can now proceed to find what the maximum possible score of nominal outlyingness is. This is derived in Proposition \ref{prop:maxdiscscore}.

\begin{proposition}\label{prop:maxdiscscore}
    Let $\maxlen = p$ and assume the minimum support threshold value is given by its maximum value so that $\sigma_d/\operatorname{supp}(d) = n-\delta$, where $\delta \rightarrow 1$ for all itemsets $d$. The maximum value of the score of nominal outlyingness $s(\boldsymbol{x}_i)$ for the $i$th observation of a data set with $p$ nominal variables is given by $p(n-1)$ for $p < 2^{r+1}+1$, and by the following expression for $p \geq 2^{r+1}+1$
    \begin{equation*}
    s^\mathrm{max}(\boldsymbol{x}_i) = (n-1) 
      \frac{\dbinom{p}{\min\left\{\maxlen, \left\lfloor \frac{p-r}{2} \right\rfloor \right\}}}{\left(\min\left\{\maxlen, \left\lfloor \frac{p-r}{2} \right\rfloor \right\}\right)^r},
    \end{equation*}
    where $\lfloor \cdot \rfloor$ is the floor function (i.e. $\lfloor x \rfloor$ is the largest integer that is smaller than or equal to $x$).
\end{proposition}
\begin{proof}
    As proven in Proposition \ref{prop:maxdiscscore_expression},
    %\textcolor{red}{if we assume that the largest possible value of $\sigma_d/\operatorname{supp}(d)$ is given by $n-1$,} 
    under the given assumptions, the score of nominal outlyingness is maximised when we only have highly infrequent (or frequent) itemsets of unit length, for $p < 2^{r+1} + 1$. In this case, it is clear that the total number of such itemsets is equal to $p$, thus the score becomes equal to $(n-1)p$.\\
    
    For the case of $p \geq 2^{r+1}+1$, as previously shown, the largest possible score can be achieved on the boundary of the solution space; that is for all itemsets of a specific length $k \leq p$ having unit support. Therefore, we seek to find
    $$k^* = \argmax\limits_{1 \leq k \leq p} \frac{C^{p}_{k}}{k^r}.$$
    The following two conditions are thus obtained by considering $k = k^*-1$ and $k = k^* + 1$
    \begin{align}
        \dbinom{p}{k^*+1}\frac{1}{\left(k^*+1\right)^r}-\dbinom{p}{k^*}\frac{1}{\left(k^*\right)^r} & < 0, \label{eq:ineq_maxdiscscore_1}\\
        \dbinom{p}{k^*}\frac{1}{\left(k^*\right)^r} - \dbinom{p}{k^*-1}\frac{1}{\left(k^*-1\right)^r} & > 0 \label{eq:ineq_maxdiscscore_2}.
    \end{align}
    Starting with Expression \eqref{eq:ineq_maxdiscscore_1}, we have
    \begin{multline*}
        \dbinom{p}{k^*+1}\frac{1}{\left(k^*+1\right)^r}-\dbinom{p}{k^*}\frac{1}{\left(k^*\right)^r} =\\
         = \frac{p!}{\left( k^* + 1 \right)!\left(p - k^* - 1\right)!}\frac{1}{\left(k^*+1 \right)^r} - \\
         - \frac{p!}{\left( k^* \right)!\left(p - k^* \right)!}\frac{1}{\left(k^* \right)^r}\\
         = \frac{p!}{\left(k^*\right)!\left(p-k^*-1\right)!\left(k^*\right)^r} \times \\
         \times \left\{\frac{1}{\left(k^*+1\right)\left(1+\frac{1}{k^*}\right)^r}-\frac{1}{p-k^*} \right\},
    \end{multline*}
    so it suffices to show that the expression inside the curly brackets is negative.
    $$\begin{aligned}
        \frac{1}{\left(k^*+1\right)\left(1+\frac{1}{k^*}\right)^r}-\frac{1}{p-k^*} =\\
        = \frac{\left(p-k^* \right)-\left( k^*+1\right)\left(1+\frac{1}{k^*}\right)^r}{\left(k^*+1 \right)\left(1+\frac{1}{k^*}\right)^r\left(p-k^*\right)}\\
    \end{aligned},$$
    so we show the numerator is negative, since the denominator is strictly positive. This gives
    \begin{align}
        & \left(p-k^* \right)-\left( k^*+1\right)\left(1+\frac{1}{k^*}\right)^r < 0 \nonumber\\
        \iff & p < k^* + \left( k^*+1\right)\left(1+\frac{1}{k^*}\right)^r. \label{eq:ineq1_maxdscore}
    \end{align}
    We proceed similarly to derive an additional bound based on Expression \eqref{eq:ineq_maxdiscscore_2}. More precisely
    \begin{multline*}
        \dbinom{p}{k^*}\frac{1}{\left(k^*\right)^r} - \dbinom{p}{k^*-1}\frac{1}{\left(k^*-1\right)^r} = \\
        = \frac{p!}{\left( k^* \right)!\left(p - k^* \right)!}\frac{1}{\left(k^* \right)^r} - \\
        - \frac{p!}{\left( k^* - 1 \right)!\left(p - k^* + 1\right)!}\frac{1}{\left(k^*-1 \right)^r}\\
        = -\frac{p!}{\left(k^*-1\right)!\left(p-k^*\right)!\left(k^*-1\right)^r}\times \\
        \times \left\{\frac{1}{p-k^*+1} - \frac{\left(1-\frac{1}{k^*} \right)^r}{k^*}\right\},
    \end{multline*}
    so since the first term is strictly positive and the product must be positive, we require that the expression inside the curly brackets is negative. Re-arranging gives
    \begin{align*}
        \frac{1}{p-k^*+1} - \frac{\left(1-\frac{1}{k^*} \right)^r}{k^*} =\\
        \frac{k^*-\left(p-k^*+1 \right) \left(1-\frac{1}{k^*} \right)^r}{\left(p-k^*+1 \right)k^*},
    \end{align*}
    thus we require that the numerator is negative, since the denominator is strictly positive. Hence, we have
    \begin{align}
        & k^*-\left(p-k^*+1 \right) \left(1-\frac{1}{k^*} \right)^r < 0 \nonumber\\
        \iff & p > \frac{k^*}{\left(1-\frac{1}{k^*}\right)^r} + k^* - 1. \label{eq:ineq2_maxdscore}
    \end{align}
    Thus, $k^*$ should be an integer such that Expressions \eqref{eq:ineq1_maxdscore} \& \eqref{eq:ineq2_maxdscore} are satisfied. We plot the above two bounds of $p$ that we obtained for varying $r$ in Figure \ref{fig:ineq_region_plots} to get an idea of what the region of solutions looks like and include a zoomed in version for the case of $r=2$ just to get a better understanding of what is going on in Figure \ref{fig:ineq_region_wzoom2}. 
    %\textcolor{red}{We will show that as $p \rightarrow \infty$, the solution can be approximated by $k^* = \lfloor \frac{p-r}{2} \rvert$, where $\lfloor \cdot \rfloor$ is the floor function. In order to motivate this, we plot the above two bounds of $p$ that we obtained for varying $r$ in Figure \ref{fig:ineq_region_plots} to get an idea of what the region of solutions looks like. We also include a zoomed in version for the case of $r=2$ just to get a better understanding of what is going on in Figure \ref{fig:ineq_region_wzoom2}.}
    \begin{figure}[h!]
        \centering \centerline{\includegraphics[width=\columnwidth]{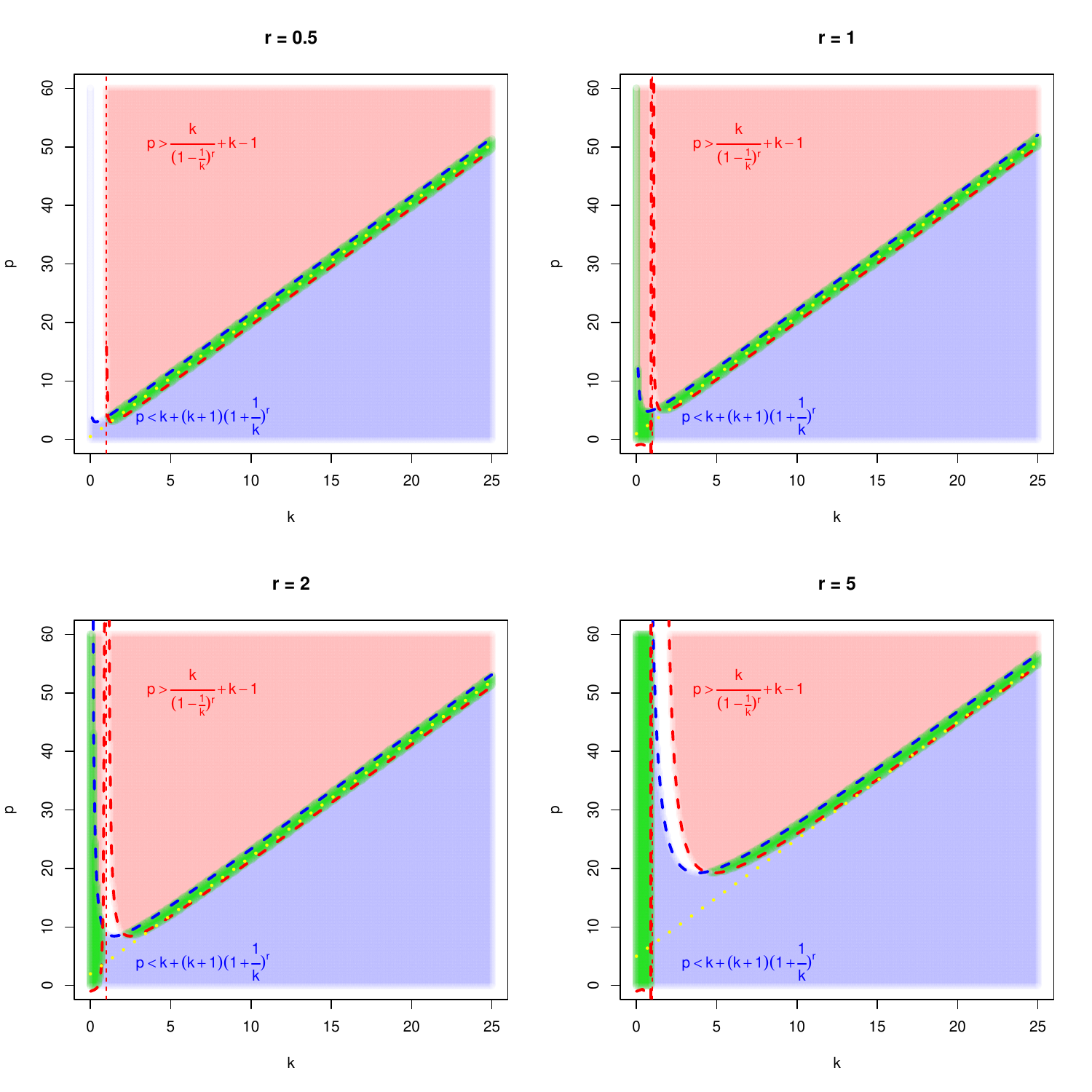}}
        \caption{Plot of bounds of $p$ for varying $r$. The blue shaded region corresponds to the upper bound from Expression \eqref{eq:ineq1_maxdscore} and the red shaded region corresponds to the lower bound from Expression \eqref{eq:ineq2_maxdscore}. The green shaded region is the region where both these bounds are satisfied. The vertical dashed line at $k=1$ is the asymptote of the expression for the lower bound. The yellow dotted line is given by $p = 2k + r$; this is the line approximating the integer solutions in the green shaded region as $p \rightarrow \infty$.}
        \label{fig:ineq_region_plots}
    \end{figure}
    \begin{figure*}[h!]
        \centering \centerline{\includegraphics[width=\textwidth]{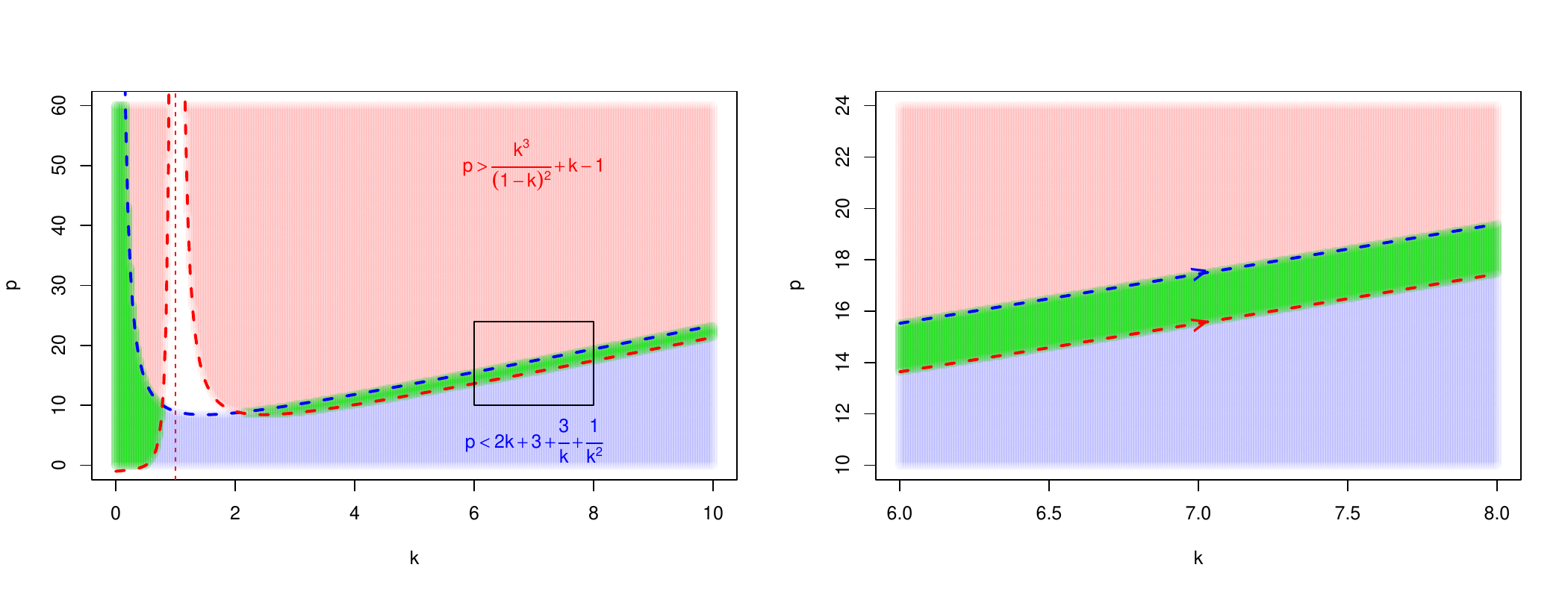}}
        \caption{Plot of bounds of $p$ for $r=2$. The rectangular region $[6,8]\times[10,24]$ on the left subplot is zoomed in on the right, illustrating how the 2 boundary curves become parallel as $p$ increases.}
        \label{fig:ineq_region_wzoom2}
    \end{figure*}
    The point of intersection of the boundary curves corresponds to $ p = 2^{r+1}+1$ and the two boundary curves become parallel as $k \rightarrow +\infty$. We validate this by computing gradients below
    \begin{multline*}
        \frac{\mathrm{d}\left( k + \left( k+1\right)\left(1+\frac{1}{k}\right)^r \right)}{\mathrm{d}k} = \\
         = 1 + \left( 1 + \frac{1}{k} \right)^r - r\frac{k+1}{k^2}\left( 1 + \frac{1}{k}\right)^{r-1} \overset{k\rightarrow \infty}{\longrightarrow} 2, \\
        \frac{\mathrm{d}\left( \frac{k}{\left(1-\frac{1}{k}\right)^r} + k - 1 \right)}{\mathrm{d}k} =\\
         = 1 + \left(1 - \frac{1}{k}\right)^{-r} - \frac{r}{k}\left(1-\frac{1}{k}\right)^{-r-1} \overset{k\rightarrow \infty}{\longrightarrow} 2.
    \end{multline*}
    Therefore, we can access the set of solutions by considering the line $p = 2k+\alpha$ and evaluating the value of $\alpha$. The difference between the two boundary curves is
    \begin{align*}
        k + \left( k+1\right)\left(1+\frac{1}{k}\right)^r - \frac{k}{\left(1-\frac{1}{k}\right)^r} + k - 1 =\\
        = \frac{(k+1)(k^2-1)^r-k^{2r+1}}{k^r(k-1)^r} + 1 \overset{k\rightarrow \infty}{\longrightarrow} 2.
    \end{align*}
    Thus, for large $k$ the vertical distance between the two boundary curves is equal to two units. Their midpoint is just one unit away from each curve, thus we require
    $$\begin{aligned}
        k + (k+1)\left(1 + \frac{1}{k}\right)^r - 2k - \alpha \overset{k\rightarrow \infty}{\longrightarrow} 1,\\
    \end{aligned}$$
    which is satisfied for $\alpha=r$. This is shown below using the binomial series expansion of $(1+1/k)^r$
    \begin{align*}
        k + (k+1)\left(1 + \frac{1}{k}\right)^r - 2k - \alpha =\\
        = k + (k+1) \left(1 + \frac{r}{k} + \mathcal{O}(k^{-2})\right) - 2k - \alpha \\
        = k + k + r + 1 + \mathcal{O}(k^{-1}) - 2k - \alpha\\
        = r + 1 - \alpha + \mathcal{O}(k^{-1}) \overset{k\rightarrow \infty}{\longrightarrow} 1,
    \end{align*}
    giving $\alpha = r$, as required. 
    %\textcolor{red}{This means that for large $k$ values, we can compute $k^*$ by accessing the region where the bounds for $p$ are satisfied via the line $p = 2k+r$. We can only use this line to determine $k^*$ as long as $p$ is large enough. Moreover,} 
    Notice that using this gives $k = (p-r)/2$, which means that for non-integer or odd integer values of $p-r$ we do not get an integer value for $k$. It is easy to show that for these cases, the only integer $k^*$ such that the bounds for $p$ are satisfied is $k^*= \lfloor (p-r)/2 \rfloor$, where $\lfloor \cdot \rfloor$ is the floor function ($\lfloor x \rfloor$ is the greatest integer $x'$ such that $x' \leq x$).\\
    
    Finally, as $\maxlen \geq 1$, the maximum score for $p < 2^{r+1} + 1$ is equal to $p$, while for $p \geq 2^{r+1} + 1$, the maximum score is equal to the expression that we mentioned at the beginning of the proof, evaluated at $k=k^*$. Since $\maxlen$ can be less than $k^*$ and given that the maximum score expression is increasing for $k \in \mathbb{Z}^+_{\leq k^*}$, the maximum score of nominal outlyingness is attained for $k = \min\{\lfloor (p-r)/2 \rfloor, \maxlen\}$.
\end{proof}

\section{Applications}\label{sec:applications}

In this section we illustrate the validity of our proposed method via a simulation study including synthetic data sets. Moreover, we apply our algorithm on four publicly available data sets from the UCI Machine Learning Repository \citep{uci} and compare its performance with that of two state-of-the-art algorithms from the frequent pattern mining literature.

\subsection{Simulation Study}

In order to illustrate the validity of our method, we conduct a simulation study on artificial data. Our proposed framework, %\textcolor{red}{which we denote by SONO (Scores Of Nominal Outlyingness)}
SONO, is benchmarked against two frequent pattern mining algorithms, specifically Frequent Pattern Outlier Factor (FPOF) \citep{he2004frequent} and Frequent Pattern Isolation (FPI) \citep{kuchar2018spotlighting}. These methods present certain similarities with SONO in the way they define outliers, yet they require specifying values for $\maxlen$ and for the minimum relative support threshold. The minimum relative support threshold is analogous to the expression $\sigma_d/n$ under our framework, where $n$ denotes the number of observations in the data.

We generate our artificial data set $\boldsymbol{X}$ using the model below
\begin{align*}
    \boldsymbol{X}_1 \sim \text{Categorical}\left(q_1, \underbrace{\frac{1-q_1}{\ell-1}, \ldots ,\frac{1-q_1}{\ell-1}}_{(\ell-1) \text{ times}}\right), \\
    \boldsymbol{X}_2 \sim \text{Categorical}\left(q_2, \underbrace{\frac{1-q_2}{\ell-1}, \ldots ,\frac{1-q_2}{\ell-1}}_{(\ell-1) \text{ times}}\right),\\
    \boldsymbol{X}_3,\ldots, \boldsymbol{X}_{p-1}, \boldsymbol{X}_p  \overset{\text{i.i.d.}}{\sim} \text{Categorical}\left( \underbrace{\frac{1}{\ell}, \ldots ,\frac{1}{\ell}}_{\ell \text{ times}}\right).
    \end{align*}
In the above, $\boldsymbol{X}_j$ denotes the $j$th variable, $p$ is the number of variables and $\ell$ is the number of categorical levels. The last $p-2$ variables are essentially sampled from a discrete uniform distribution, while the first level for the first two variables occurs with probabilities $q_1$ and $q_2$, respectively. We set $q_1=0.10$ and $q_2=0.05$ in order to have one level being significantly less frequent than the rest. We vary the values of $n$, $p$ and $\ell$ and allow these to take values $n = 200, 500, 1000$, $p = 3, 5, 7$, and $\ell = 2, 3, 5$. This leads to 27 different scenarios, simulated a hundred times each (using a different seed number), thus yielding 2700 artificial data sets.

For each data set, we run FPOF and FPI with a minimum relative support threshold of $1/\ell$ and $\maxlen$ determined by our method using the relevant functions from the \texttt{R} package \texttt{fpmoutliers}. SONO is run using the \texttt{R} package \texttt{SONO} \citep{costa_sono_2025} three times; once assuming equal probabilities across all levels for all variables, once with the true probabilities, and once with probability vectors $\left(q_j/2, (1-q_j/2)/(\ell-1), \ldots, (1-q_j/2)/(\ell-1)\right)^\intercal$,  $j=1, 2$, for variables $\boldsymbol{X}_1, \boldsymbol{X}_2$ and equal probabilities across each level for the remaining $p-2$ variables. We record the average rank for the observations possessing the first level for any of $\boldsymbol{X}_1$ or $\boldsymbol{X}_2$, the estimated value of $\maxlen$ and the average runtime for each of the three algorithms across twenty-five repetitions. Since some observations may end up with identical scores, we use the highest rank to resolve any ties (a higher rank here refers to a larger score of outlyingness). We also subtract the scores obtained with FPOF from a unit for their ranks to be comparable to these computed using SONO and FPI (FPOF assigns higher scores to inliers). We finally set $r=1$ to be more consistent with FPOF and FPI and $\alpha = 0.05$ when running SONO.

In Figure \ref{fig:medrank_nobs_nvars} we can see the median rank for observations for which $\boldsymbol{X}_1 = 1$ or $\boldsymbol{X}_2=1$, as obtained by FPI, FPOF and SONO for a varying number of variables and observations. For the latter method, the three different cases related to the assumed probabilities are also included. When SONO is given the true probabilities, the median rank is constantly equal to a unit; this comes as no surprise since all scores of outlyingness are equal to zero, thus all observations get the exact same rank of one. Incorporating this contextual information and thus not flagging any observation as potentially outlying is not possible with FPI and FPOF, with both algorithms consistently treating the observations of interest as potential outliers. SONO has also managed to flag these and has even assigned them a higher rank than FPI and FPOF under the assumption of equal level probabilities. On the other hand, when the assumed probabilities for the first levels of $\boldsymbol{X}_1$ and $\boldsymbol{X}_2$ are smaller than the true ones, we see that the median rank is much lower and that is because the concept of infrequent itemsets is no longer applicable, further highlighting the validity of our method.

We finally report the median runtime in seconds for our experiments, which heavily depends on the value $\maxlen$, in Table \ref{tab:maxlen_tab}. The reported times are only for the case of probability levels being set equal across each variable to ensure a fair comparison between the three methods. As we can see, SONO is quicker than FPI and FPOF when $\maxlen=1$ and consistently faster than FPI for $\maxlen \leq 3$. The median runtime of SONO increases on a large scale as $\maxlen$ grows (potentially due to its current implementation not being fully optimised for speed), yet it does not take extremely long to produce results which are typically more reliable than these of FPOF and FPI, as seen previously. Finally, we notice that a $\maxlen$ value of seven is only attained in 181 out of the 900 data sets consisting of seven variables and only when the true underlying level probabilities are provided. This highlights the efficiency of using support-based pruning, leading to a substantially smaller runtime in approximately 80\% of the data sets with seven variables.

\begin{figure*}[h!]
  \centering
  \includegraphics[width=\textwidth]{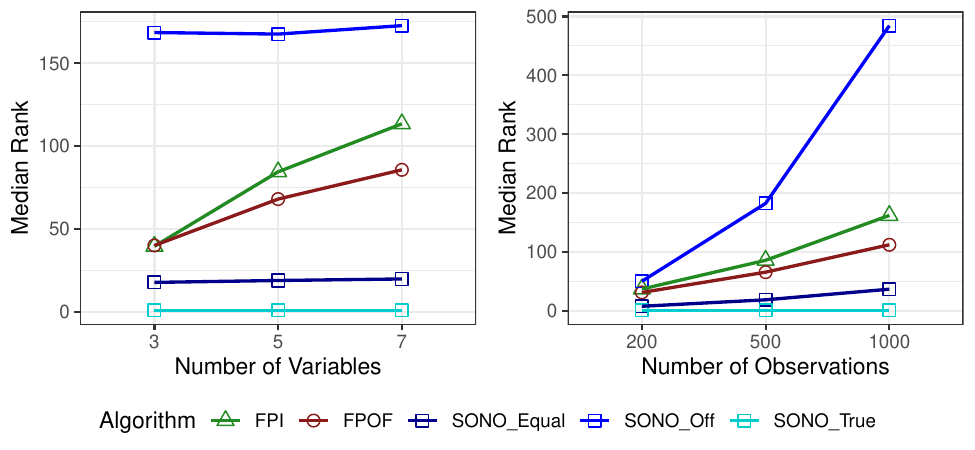}
  \caption{Median rank of observations for which $\boldsymbol{X}_1 = 1$ or $\boldsymbol{X}_2=1$ using FPI, FPOF and SONO with equal probabilities across levels, the true level probabilities and with half the true probabilities for the first levels of $\boldsymbol{X}_1, \boldsymbol{X}_2$.}
  \label{fig:medrank_nobs_nvars}
\end{figure*}
\begin{table}[h!]
\begin{tabular}{r|*{3}{c}}
\backslashbox{MAXLEN}{Algorithm}
& SONO & FPOF & FPI \\\hline
1 & \textbf{0.03} & 0.04 & 1.18\\
2 & 0.67 & \textbf{0.08} & 1.19 \\
3 & 0.54 & \textbf{0.08} & 1.18 \\
4 & 1.69 & \textbf{0.04} & 1.02 \\
5 & 3.75 & \textbf{0.08} & 1.19 \\
6 & 24.2 & \textbf{0.14} & 1.29 \\
\end{tabular}
\caption{Median runtime in seconds for SONO, FPI and FPOF for MAXLEN values of one up to six, assuming equal probabilities across variable levels. Smallest times per $\maxlen$ value are bolded.}
\label{tab:medruntime_maxlen}
\end{table}
\subsection{Applications on publicly available data}
We implement our method and compare its performance to that of FPOF and FPI on four publicly available data sets from the UCI Machine Learning Repository. The data sets used are the Solar Flare, the Differentiated Thyroid Cancer Recurrence, the Lymphography, and the Wisconsin Breast Cancer data sets. These data sets consist of mainly nominal variables and the few variables of a different type (if any) are removed in the cleaning process, so that our method can be applied. Notice that these data sets have previously been used for classification purposes \citep[see for example][]{Michalski1986TheMI, Clark1987InductionIN, bradshaw1992forecasting, akay2009support, borzooei2024machine}. The Lymphography and Wisconsin Breast Cancer data sets have also been used for anomaly detection purposes in multiple instances \citep[for instance in][among others]{he2004frequent, lazarevic2005feature, zimek2013subsampling, aggarwal2015theoretical, kuchar2018spotlighting}.

The number of observations for each data set upon cleaning, as well as the number of variables used for each data set are summarised in Table \ref{tab:maxlen_tab}, together with the value of $\maxlen$ as determined by SONO. Notice that we set $\alpha = 0.05$ and $r = 1$ for SONO, while the probabilities are set to be equal across variable categories. Following the implementations of FPOF and FPI on the Lymphography and Wisconsin Breast Cancer data sets from the papers that introduced these algorithms, we set the minimum relative support threshold to be $0.1$. The value of $\maxlen$ is set to be the same as the one estimated by SONO for each data set. We also report the proportion of outliers as a percentage for each data set, assuming that instances from the minority class are outlying for the Thyroid Cancer, Breast Cancer and Lymphography data sets. For the Solar Flare data set, observations corresponding to moderate or severe solar flares are taken to be the outliers.

\begin{table*}[h!]
\centering
\begin{tabular}{l|c|c|c|c}
Data set & Observations & Variables & $\maxlen$ & Proportion of Outliers (\%)\\ \hline
Solar Flare & 1389 & 10 & 6 & 5.26\%\\
Thyroid Cancer & 383 & 13 & 4 & 28.20\%\\
Wisconsin Breast Cancer & 683 & 9 & 1 & 34.99\%\\
Lymphography & 148 & 18 & 4 & 4.05\%\\
\end{tabular}
\caption{Number of observations upon data cleaning, number of variables used for nominal outlier detection, $\maxlen$ values and proportion of outliers for publicly available data sets used.}
\label{tab:maxlen_tab}
\end{table*}
Following the recommendations of \cite{campos2016evaluation}, we compare the performance of the three algorithms using three evaluation measures. The first one we use is the Detection Rate at Top $K\%$, defined as the proportion of outliers detected within the observations with the top $K\%$ of outlyingness scores. This is analogous to the Recall measure, so we denote this by $\mathcal{R}(K)$. Assuming $O$ is the set of outliers, $\mathcal{R}(K)$ is defined as follows
$$\mathcal{R}(K) = \frac{ \lvert \{o \in O: \text{rank}(o) \leq \lceil Kn/100\rceil \} \rvert}{\lvert O \rvert},$$
where $\text{rank}(o)$ denotes the rank of an outlying instance $o \in O$ in a sorted list of outlyingness scores (ranked in descending order, with a rank of one being the most outlying).

We also compute the average rank for the outlying observations and Area Under the Receiver Operating Characteristic Curve (ROC AUC) value for each method. The ROC AUC value for a list of ranked scores is computed as follows
\begin{align*}
    \text{ROC AUC} & = \frac{1}{\lvert O \rvert\times \lvert I\rvert}\sum_{o \in \lvert O \rvert}\sum_{i \in \lvert I \rvert}\mathcal{S}(o, i), \\
    \mathcal{S}(o, i) & = \left\{
    \begin{array}{ll}
          1, & \text{if rank}(o) > \text{rank}(i) \\
          \frac{1}{2}, & \text{if rank}(o) = \text{rank}(i)\\
          0, & \text{if rank}(o) < \text{rank}(i) \\
    \end{array} 
    \right. ,
\end{align*}
where $I$ is the set of inliers \citep{hanley1982meaning}. A value closer to one indicates perfect performance, as this implies that most outliers are ranked higher than most of the inliers. In contrast, a value closer to zero occurs when an algorithm assigns higher outlyingness scores to inliers. We report $\mathcal{R}(K)$, the average rank and the ROC AUC for the $K\%$ top-ranked observations for each data set and each algorithm in Table \ref{tab:avg_ranks} and Figures \ref{fig:Rk_uci} \& \ref{fig:rocauc_topK_UCI} for several values of $K$.

\begin{table}[h!]
\centering
\begin{tabular}{l|c|c|c}
Data set & SONO & FPOF & FPI\\ \hline
Solar Flare & 377.48 & 348.34 & \textbf{321.84} \\
Thyroid Cancer & 101.74 & \textbf{92.04} & 99.45\\
Wisconsin Breast Cancer & \textbf{116.63} & 125.67 & 126.08 \\
Lymphography & \textbf{4.00} & 4.50 & \textbf{4.00} \\
\end{tabular}
\caption{Average rank of outlying observations in four UCI data sets based on the scores of outlyingness computed using the SONO, FPOF and FPI algorithms. Values are presented in two decimal places and a lower value indicates a higher ranking. Highest average rankings for each data set are bolded.}
\label{tab:avg_ranks}
\end{table}

\begin{figure*}[h!]
  \centering
  \includegraphics[width=\textwidth]{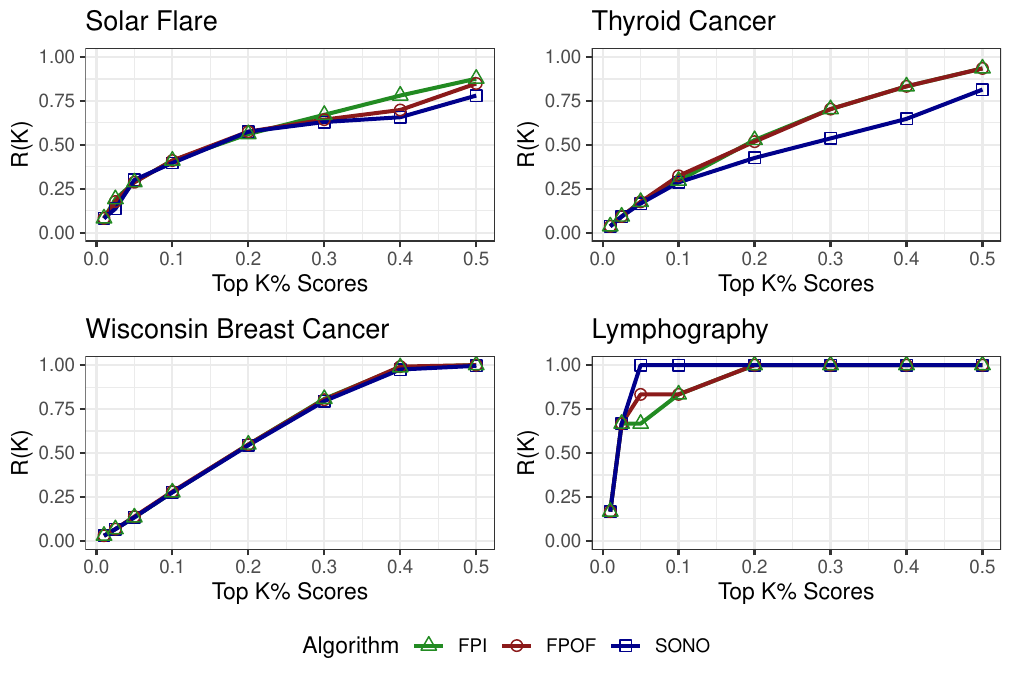}
  \caption{Detection Rate at Top $K\%$ ($\mathcal{R}(K)$) for observations with the top $K\%$ scores of outlyingness for FPI, FPOF and SONO algorithms on the Solar Flare, Thyroid Cancer, Wisconsin Breast Cancer and Lymphography data sets. The values of $K$ presented are 1, 2.5, 5 and values from 10 up to 50, in steps of 10.}
  \label{fig:Rk_uci}
\end{figure*}

\begin{figure*}[h!]
  \centering
  \includegraphics[width=\textwidth]{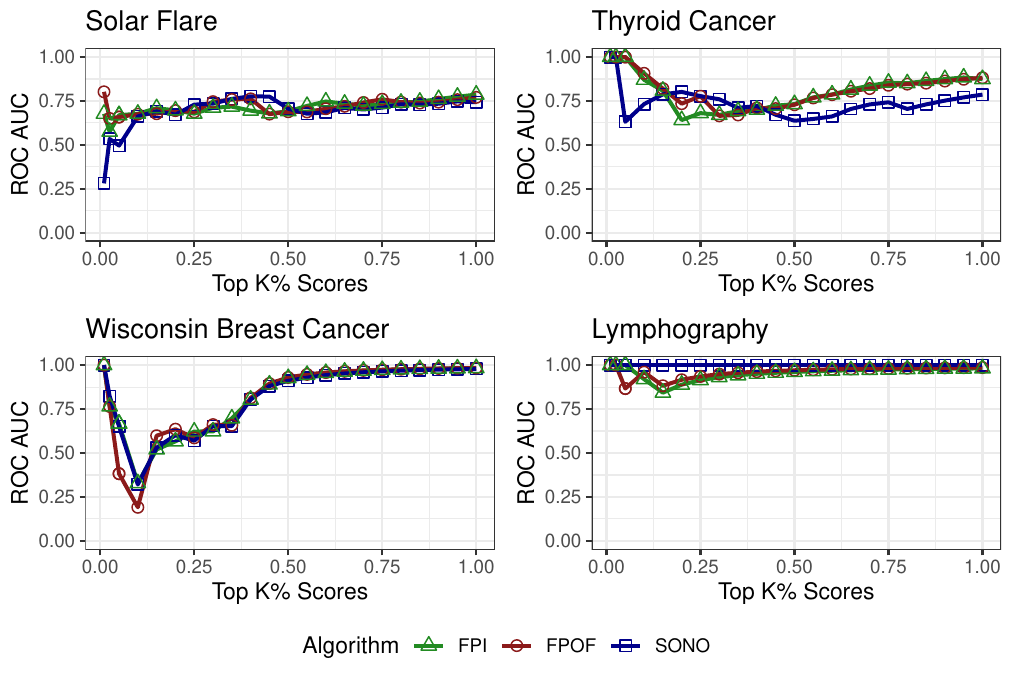}
  \caption{ROC AUC for observations with the top $K\%$ scores of outlyingness for FPI, FPOF and SONO algorithms on the Solar Flare, Thyroid Cancer, Wisconsin Breast Cancer and Lymphography data sets. The values of $K$ presented are 1, 2.5 and values from 5 up to 100, in steps of 5.}
  \label{fig:rocauc_topK_UCI}
\end{figure*}

Based on the average rank of outliers in Table \ref{tab:avg_ranks}, we see that SONO outperforms FPOF and FPI on the Wisconsin Breast Cancer data, while it achieves better performance than FPOF on the Lymphography data set. The $\maxlen$ values for Thyroid Cancer, Wisconsin Breast Cancer, and Lymphography are also considerably smaller than the number of variables. This makes the search for infrequent itemsets less computationally intensive without necessarily affecting performance.

Despite not being able to do as well as its competitors on the Solar Flare and the Thyroid Cancer data sets, SONO has managed to recover a good amount of outliers, as can be seen in Figure \ref{fig:Rk_uci}. In fact, for the Solar Flare data set, we see that the Detection Rate $\mathcal{R}(K)$ is comparable to that of FPOF or FPI for $K \leq 30\%$. 
%\textcolor{red}{Similarly for the Thyroid Cancer Data set, the largest scores of outlyingness obtained with each of the three methods correspond in their majority to outliers.} 
We also observe an almost identical performance of the three methods on the Wisconsin Breast Cancer data set, whereas SONO outperforms its competitors on Lymphography.

The ROC AUC for SONO takes very similar values to the ones obtained with FPI and FPOF, as seen in Figure \ref{fig:rocauc_topK_UCI}. A constant ROC AUC score of a unit is achieved by SONO on the Lymphography data set, indicating that all outliers in the data set received higher scores than the inliers. Finally, we highlight that in a practical application, where the proportion of outliers is typically assumed to be less than $50\%$, SONO is a competitive algorithm to the ones presented, being able to correctly detect over 70\% of the outliers.

Overall, SONO exhibits comparable performance to FPI and FPOF. However, it is important to stress that SONO requires far less user input than its competitors, since the value for $\maxlen$ is estimated from the data and none of the minimum support threshold values are input by the user. Different values for the minimum relative support threshold may produce much different results for FPI and FPOF, yet these are not explored here and in fact, no recommendations are made regarding the selection of an optimal value for this hyperparameter.

A final appealing feature of SONO is that we can interpret the scores obtained and see which variables have contributed to these using the additional concepts introduced in Section \ref{sec:proposal}. 
%\textcolor{red}{We will not delve into the details of the itemsets that contributed in the scores presented for each data set but we will give some basic guidelines on how the contribution matrix $\boldsymbol{\mathrm{{C}}}$ and the Nominal Outlyingness Depth can facilitate their interpretation.} 
For instance, the contribution matrices for the outliers of each data set are illustrated in Figure \ref{fig:contrib_mats_uci}. As can be seen, the size of the sun spot area or whether there is any activity in that region in the last twenty-four hours seem to be the most influential variables for a solar flare to be flagged as possibly moderate or severe. Similarly, the stage, metastasis (M) classification, and whether the patient has received radioactive iodine therapy are indicators of thyroid cancer recurrence, while marginal adhesion and mitoses can be used to detect cases of malignant breast tumor. Finally, the rate of lymphatic diminishment and the blockage of lymphatic sinuses can hint normal finds or cases of fibrosis in the Lymphography data set.

We also plot the Nominal Outlyingness Depth against the scores for the outliers in the Solar Flare, Thyroid Cancer and Lymphography data sets, for which $\maxlen$ is larger than one in Figure \ref{fig:nod_sono_flare}. Some observations have a zero depth and score, while most outliers have at least a depth of one and one observation even has a depth of two. 
%\textcolor{red}{The largest score observed for the Solar Flare data set corresponds to a depth of nearly two units, suggesting that contributions from itemsets of both unit and greater lengths have resulted in this substantial score, making this specific observation a probable outlier.}
The estimated values of $\maxlen$ are much higher for all three data sets than the maximum nominal outlyingness depths attained. This hints that the longer itemsets do not contribute in the calculation of the scores; e.g. in the Solar Flare data set, where the maximum depth is equal to a unit. This highlights a drawback of our strategy for selecting $\maxlen$, as calculations for itemsets of length two up to six on Solar Flare are eventually redundant, at least for distinguishing the outliers from the rest of the observations.

\begin{figure}[h!]
  \centering
  \includegraphics[width=\linewidth]{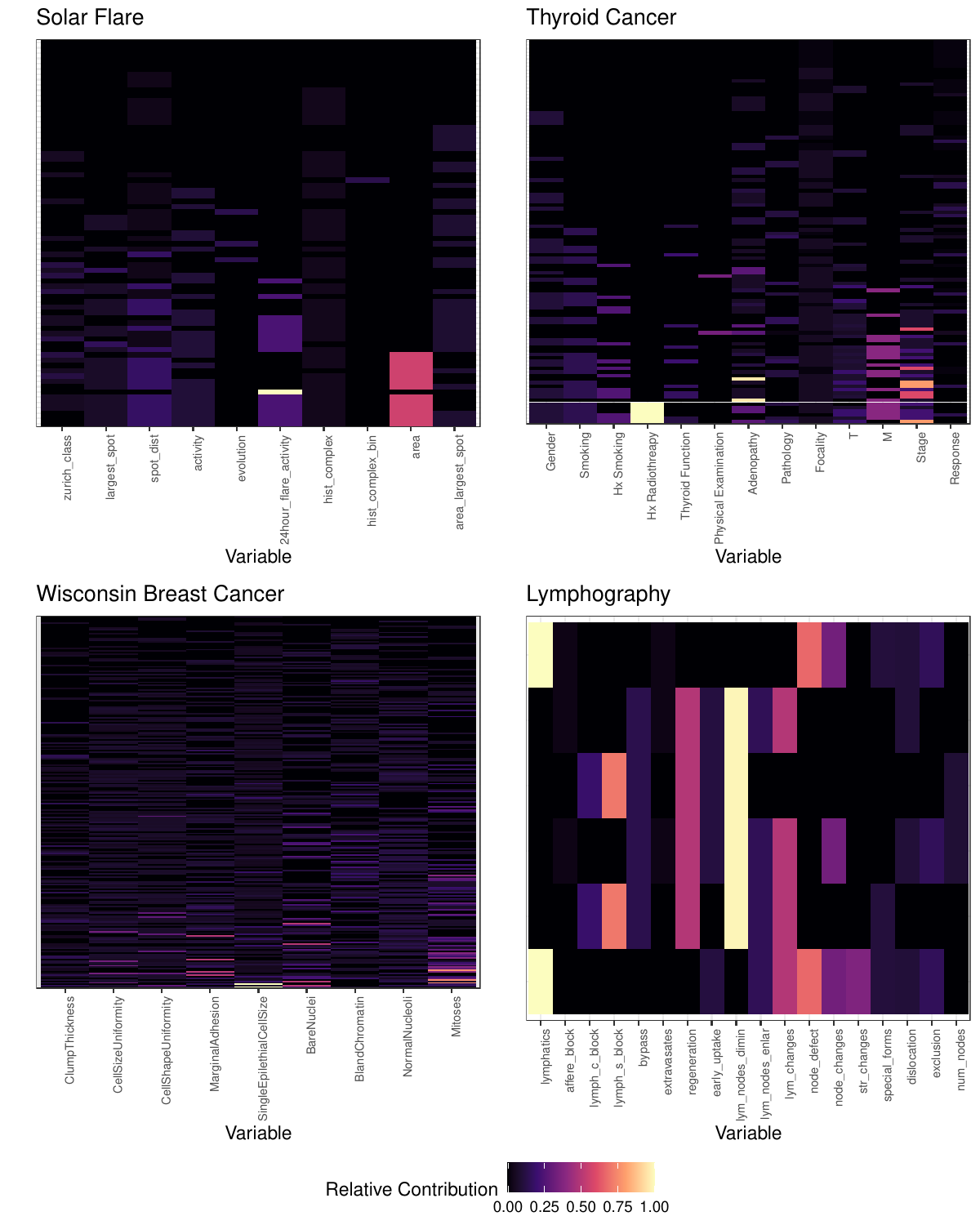}
  \caption{Contribution matrices for UCI data sets; the contributions have been scaled to lie within the unit interval and the relative contribution is presented. Only the outlying observations are presented.}
  \label{fig:contrib_mats_uci}
\end{figure}

\begin{figure}[h!]
  \centering
  \includegraphics[width=\columnwidth]{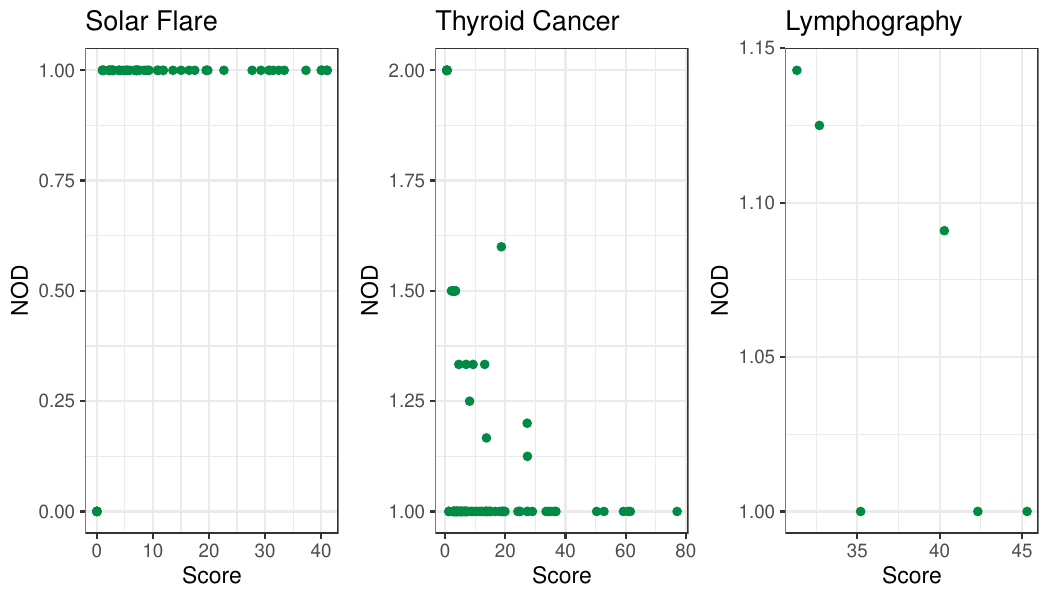}
  \caption{Nominal Outlyingness Depth (NOD) against scores of nominal outlyingness computed with SONO for the outliers of the Solar Flare, Thyroid Cancer and Lymphography data sets.}
  \label{fig:nod_sono_flare}
\end{figure}

\section{Conclusion}\label{sec:conclusion}

In this paper, we propose a novel framework for quantifying outlyingness in data sets consisting of nominal variables. We define nominal outlyingness as the degree to which a sequence of nominal values and its subsets differ from a Multinomial distribution with varying event probabilities that is assumed to have generated them. Based on this definition, one may define outliers as observations consisting of highly frequent of highly infrequent sequences and supersets or subsets of sequences, respectively.

Under the framework implied by the proposed definition of nominal outlyingness, we introduce a method for quantifying outlyingness for nominal data. Our proposal includes ideas from the association rule mining literature and can be seen as an extension of the work of \cite{koufakou2010}. We present a strategy that can be used to determine values for the hyperparameters that are involved in the calculation of the score of nominal outlyingness and provide details on support-based pruning, which seeks to reduce computational complexity. Furthermore, we define the concept of nominal outlyingness depth and the contribution matrix. These can be used to assess the source of nominal outlyingness by providing the user with the average length of highly frequent/infrequent sequences included in an observation, as well as with the extent to which each nominal variable has contributed to the score of outlyingness of an observation.

Our proposed framework is compared to two state-of-the-art algorithms for nominal outlyingness quantification, namely FPOF and FPI. The efficacy our method is illustrated via a series of simulations on 2700 synthetic data sets, where the results align with what is expected given the data generating process. In contrast, the two competing algorithms act in a completely agnostic manner, inflating the scores for observations with nominal levels that are not genuinely outlying but are set by default to appear fewer times than others. We apply our method on four publicly available data sets. Two of these data sets have not been previously used for testing outlier detection but rather for classification purposes, while the other two typically serve as benchmarking data sets for newly developed outlier detection algorithms. Our method produces comparable and in certain instances even better results to these of FPOF and FPI, under the assumption of the minority class being the class of outliers in these data sets. Moreover, our proposed framework selects suitable hyperparameter values without requiring user input, a feature that is not supported by the other two methods.

The proposed framework yields some remarkable results that can be easily interpreted (thanks to the introduction of the nominal outlyingness depth and the contribution matrix) on the publicly available data. However, it does not come with no limitations. A shortcoming of the framework presented is the way the hyperparameter $\maxlen$ is computed. For instance, if we are interested in highly infrequent itemsets and we have a large number of nominal variables with highly imbalanced level proportions, the expected number of sequences containing only the most frequent levels is typically large enough to produce sparse expected contingency tables. A direct consequence of this is that $\maxlen$ becomes very large, leading to potentially redundant calculations and therefore an increased computational cost. This is observed in some cases in our simulations, where the depth is much lower than the computed value of $\maxlen$, indicating that sequences of greater length have limited to no contribution to the obtained scores of nominal outlyingness. An amendment in the way $\maxlen$ is determined can alleviate this issue and reduce the computational complexity, without changing the way the results are interpreted and the conclusions derived. Additionally, one may wish to extend the simulation study by enforcing a dependence structure among variables which are likely to be associated in a specific manner. This way, we can get rid of the independence assumption (which is typically not valid) and potentially obtain more interesting results.

We conclude by giving some recommendations on how our framework may be used for extending the research that is currently being conducted in certain fields. The proposed framework can be used for assigning weights to observations in a data set prior to running a classification algorithm. As an example, one may perhaps want to downweigh subjects with a large score of nominal outlyingness, as these can be seen as potential outliers which may distort the classification output. This can be done by modifying the relevant loss function to account for observation weights. Similarly, one might wish to implement such an approach for clustering categorical data or deriving a version of the Trimmed K-Means \citep{cuesta1997trimmed} algorithm for nominal data. Finally, the proposed framework can contribute in the development of more reliable predictions in linear regression with nominal predictors 
%\textcolor{red}{. Extremities are commonly the reason behind biased parameter estimates and misleading predictions of a regression model; being able to} 
by diminishing the impact of potential outliers, yielding robust regression algorithms for nominal data. These recommendations can pave the way for future research opportunities and culminate in the development of novel methodologies which in turn provide insight on research questions of interest.

\backmatter

\subsubsection*{Supplementary information}

The \texttt{R} code for the simulations conducted is available on \texttt{GitHub} at \url{https://github.com/EfthymiosCosta/SONO_Sims}.

\section*{Declarations}

\subsubsection*{Conflict of interest}
The authors declare that there are no conflict of interest to disclose.

\subsubsection*{Funding}
The first author gratefully acknowledges funding provided by EPSRC's Centre for Doctoral Training in Modern Statistics and Statistical Machine Learning grant EP/S023151/1.

%%=============================================%%
%% For submissions to Nature Portfolio Journals %%
%% please use the heading ``Extended Data''.   %%
%%=============================================%%

%%=============================================================%%
%% Sample for another appendix section			       %%
%%=============================================================%%

%% \section{Example of another appendix section}\label{secA2}%
%% Appendices may be used for helpful, supporting or essential material that would otherwise 
%% clutter, break up or be distracting to the text. Appendices can consist of sections, figures, 
%% tables and equations etc.

%%===========================================================================================%%
%% If you are submitting to one of the Nature Portfolio journals, using the eJP submission   %%
%% system, please include the references within the manuscript file itself. You may do this  %%
%% by copying the reference list from your .bbl file, paste it into the main manuscript .tex %%
%% file, and delete the associated \verb+\bibliography+ commands.                            %%
%%===========================================================================================%%

\bibliography{sn-bibliography}% common bib file
%% if required, the content of .bbl file can be included here once bbl is generated
%%\input sn-article.bbl

\end{document}